\documentclass[a4paper,12pt,english,oneside,final]{article}
\usepackage[utf8]{inputenc} 
\usepackage[T1]{fontenc} 
\usepackage[english]{babel}
\usepackage{setspace}

\usepackage{lmodern} 
\usepackage{stackengine} 
\usepackage[natbibapa]{apacite} 

\usepackage{amsthm,amssymb,mathrsfs,amsmath,amsfonts}
\usepackage{mathtools}
\usepackage{thmtools,thm-restate} 
\usepackage{euscript} 
\usepackage{verbatim}
\usepackage{caption}
\usepackage{subcaption}
\usepackage{enumitem}
\usepackage{microtype}
\usepackage{dsfont}
\usepackage{csquotes}
\usepackage{xcolor}
\usepackage{breqn}
\usepackage{tabulary}
\usepackage{booktabs}
\usepackage{graphicx}

\definecolor{coolblack}{rgb}{0.0, 0.18, 0.39}

\usepackage[colorlinks=true,linkcolor=coolblack,
urlcolor=coolblack,
citecolor=coolblack,pdftex]{hyperref} 
\usepackage[all]{hypcap}  

\usepackage{geometry}
\makeatletter
\newtheorem{theorem}{Theorem}
\newtheorem{proposition}{Proposition}

\newtheorem{Acorollary}{Corollary}[section] 

\newtheorem{Alemma}{Lemma}[section] 
\newtheorem{definition}{Definition}

\newtheorem{step}{Step}

\usepackage[shadow,colorinlistoftodos]{todonotes}

\newcommand{\p}{\text{{\footnotesize \textcircled{+}}}}
\newcommand{\n}{\text{{\footnotesize \textcircled{-}}}}
\newcommand{\dm}{dm} 
\newcommand{\ubar}[1]{\text{\b{$#1$}}}

\makeatother

\usepackage{babel}

\DeclareMathOperator*{\argmax}{arg\,max}

\begin{document}
	
	\pagestyle{plain}
	
	\author{Federico Vaccari\thanks{Laboratory for the Analysis of Complex Economic Systems, IMT School of Advanced Studies, 55100, Lucca, Italy. E-mail: \href{mailto:vaccari.econ@gmail.com}{\sf vaccari.econ@gmail.com}. Santiago Oliveros provided invaluable guidance, support, and help. I am grateful to Luciano Andreozzi, Paulo Barelli, Ennio Bilancini, John Duggan, Srihari Govindan, Mathijs Janssen, Tasos Kalandrakis, David K. Levine, Elena Manzoni, Marco Ottaviani, and Andriy Zapechelnyuk. I thank seminar audiences at Collegio Carlo Alberto, Institute for Microeconomics at the University of Bonn, University of Florence, GRASS XV, 62nd Annual Conference of the Italian Economic Association, and Formal Theory Virtual Workshop. I also thank the Editor, the Associate Editor, and three referees for their constructive comments. Previous versions of this paper circulated under the title ``Competition in Signaling.'' All errors are mine. This project has received funding from the European Union's Horizon 2020 Research and Innovation Programme (Marie Sk\l odowska-Curie grant no.~843315-PEMB).}}
	\title{Competition in Costly Talk}
	\date{}
	\maketitle

	\begin{abstract}		
		\noindent This paper studies a communication game between an uninformed decision maker and two perfectly informed senders with conflicting interests. Senders can misreport information at a cost that increases with the size of the misrepresentation. The main results show that equilibria where the decision maker obtains the complete-information payoff hinge on beliefs with undesirable properties. The imposition of a minimal and sensible belief structure is sufficient to generate a robust and essentially unique equilibrium with partial information transmission. A complete characterization of this equilibrium unveils the language senders use to communicate.
	\end{abstract}
	
	
	
	\thispagestyle{empty}
	
	\clearpage

	\tableofcontents
	
	\pagebreak

\section{Introduction}

How and how much information is revealed when two equally informed senders with conflicting interests provide advice to a decision maker? When senders are well informed and misreporting is prohibitively expensive, the decision maker can ``rely on the information of the interested parties'' to always make the right choice.\footnote{See, e.g., \cite{milgrom1986}.} However, there are many situations where information is not fully verifiable and it is possible to misreport it at a reasonable cost.\footnote{Misreporting information is a costly activity due to the time and effort that is required to misrepresent the information, or due to the expected loss in reputation, credibility, and future influence, and more. Moreover, misreporting is more difficult, and hence more costly, when information is harder.} Intuition would suggest that, in these cases, the decision maker might obtain conflicting advice and make wrong choices as a result of being poorly informed.

This type of interaction is at the core of a large number of real-world scenarios: candidates competing for consensus during an electoral campaign may provide voters with different accounts of the same events; 
newspapers with opposed political leanings may deliver conflicting and inaccurate news; prosecutors and defendants trying to persuade a jury may tamper with evidence; co-workers competing for a promotion may exaggerate their stated contributions to a team project; and methods used in lobbying against public health may include ``industry-funded research that confuses the evidence and keeps the public in doubt'' \citep{chanwho}.

I address the above questions with a costly signaling game between an uninformed decision maker and two perfectly informed senders with conflicting interests. The two senders observe the realization of a random variable---the state---and then privately deliver a report to the decision maker. These reports are literal statements about the realized state. Senders can misreport such information, but incur misreporting costs that are increasing in the magnitude of misrepresentation. By contrast, reporting truthfully is costless. After observing the reports, the decision maker must select one of two alternatives. At the end of the game, each player obtains a payoff that depends on the realized state and on the alternative selected by the decision maker.

The paper focuses on the attainment of \emph{receiver-efficient} outcomes, whereby the decision maker achieves her full-information payoff. The main results concern the robustness of these outcomes and the characterization of sensible equilibria. First, I show that the existence of receiver-efficient equilibria hinges on beliefs with undesirable properties. In particular, the decision maker must believe that deviations from the equilibrium path occur because a sender purposefully ``shot himself in the foot.'' Second, I impose minimal and natural restrictions on the decision maker's beliefs to study sensible equilibria. In doing so, I study the amount of information that can be transmitted in equilibrium and the language used by senders to deliver such information. I find that receiver-efficiency is no longer an equilibrium outcome once we impose a sensible belief structure.

Two well-known refinements eliminate receiver-efficient and pure-strategy equilibria: unprejudiced beliefs \citep{bagwell1991} and $\varepsilon$-robustness \citep{battaglini2002}.\footnote{See Section~\ref{sec:fre} for a formal definition of unprejudiced beliefs and $\varepsilon$-robustness. I show that these two refinements are tightly connected: in this model, equilibria that are $\varepsilon$-robust must be supported by unprejudiced beliefs (Lemma~\ref{lemma:epsilonunprejudiced}). This result suggests a novel rationale for the use of $\varepsilon$-robustness in multi-sender communication games.} This result motivates the search for mixed-strategy equilibria that are robust to such refinements. To conduct this search, I proceed by imposing two natural restrictions on the decision maker's posterior beliefs. Roughly, the first restriction is a monotonicity condition under which the decision maker interprets some higher reports as originating from higher states. The second restriction is a dominance condition under which the decision maker rules out that senders deliver strictly dominated reports.\footnote{See Definition~\ref{def:directeqa} in Section~\ref{sec:solutions} for a complete and formal statement of these two conditions.} These restrictions are natural given that reports are literal and misreporting is costly. I refer to equilibria satisfying these two conditions as \emph{adversarial equilibria}.

I then provide a complete characterization of adversarial equilibria, and show that, in the game considered here, they possess desirable properties: they always exist, they are essentially unique, and they withstand the refinement criteria that break down receiver-efficient and pure-strategy equilibria. The two conditions imposed on adversarial equilibria, although relatively natural and mild, are sufficient to ensure robustness and uniqueness while preserving existence.

The senders' behavior in an adversarial equilibrium is mixed, as they report the truth with some probability and misreport otherwise. Upon observing two conflicting reports recommending different actions, the decision maker understands that ``the truth is somewhere in between'' and that at least one of the two senders is misreporting. In these cases, the decision maker allocates the burden of proof between the senders in a way that depends on their characteristics, such as their bias and cost structure. Because of the senders' mixed behavior, the decision maker cannot always glean the information she needs by inverting the senders' reports. As a result, information transmission is only partial and persuasion takes place with positive probability.

The setting studied in the main part of the paper is sufficiently rich to draw general conclusions about the main setup. To illustrate the takeaways, I also analyze the specific case where senders have a similar payoff and cost structure, and where the distribution of the state is such that no sender has an ex-ante advantage of any kind. In this symmetric environment, I provide a closed-form solution to adversarial equilibria and show that they naturally display symmetric strategies. The decision maker equally allocates the burden of proof between the senders by following the recommendation of the sender delivering the most extreme report. The senders' misreporting behavior depends on the shape of the common cost function: with convex costs, senders are more likely to convey large misrepresentations of the state rather than small lies, while the opposite is true for concave misreporting costs.

This paper shows that introducing misreporting costs in a model with multiple senders has crucial implications on information transmission. When ``talk is cheap'' and misreporting is costless, no information can be transmitted if senders have conflicting interests. When information is verifiable and misreporting is impossible, only fully revealing equilibria in truthful strategies ensue. By contrast, here there are neither babbling nor truthful equilibria. Adversarial equilibria are not receiver-efficient, although they feature a probabilistic revelation of almost every state. Competition between senders also impacts how communication takes place. In related single-sender models with costly talk, there are fully revealing equilibria where senders play pure strategies. By contrast, here pure-strategy and fully revealing equilibria are not robust. In this model, the transmission of information takes place differently than in comparable models of strategic communication.

The remainder of this article is organized as follows. In Section~\ref{sec:literature}, I discuss the related
literature. Section~\ref{sec:model} introduces the model, which I solve in Sections~\ref{sec:fre} and \ref{sec:mixed}. In Section~\ref{sec:example}, I
provide an example and discuss several extensions of the baseline model. Finally, Section~\ref{sec:conclusion} concludes.
Formal proofs are relegated to Appendix~\ref{sec:app} and \ref{sec:suppapp}.


\section{Related Literature}\label{sec:literature}

This paper relates to models of strategic communication with multiple senders. This line of work shows several channels through which full information revelation can be obtained \citep{battaglini2002,krishna2001exp,milgrom1986}. Papers in this literature typically assume that misreporting is either costless (cheap talk) or impossible (verifiable disclosure). By contrast, in this article misreporting is possible at a cost that depends on the magnitude of misrepresentation. Under this modeling specification, there are no babbling equilibria or equilibria in truthful strategies. I show that fully revealing and non-revealing equilibria co-exist, but only the latter survive natural restrictions in the decision maker's beliefs.

This paper also relates to models of strategic communication with misreporting costs. \cite{kartik2007} show that there exist fully revealing equilibria of single-sender settings with an unbounded state space. In these equilibria, the receiver can always infer the true state by inverting the sender's inflated language. By contrast, there are no fully revealing equilibria of single-sender settings with a bounded state space \citep{ottaviani2006,chen2008,kartik2009,chen2011}. I extend the analysis of communication with misreporting costs to a setup with two competing senders. In this setting, I show that non-revealing equilibria exist and are robust even when the state space is arbitrarily large. In these equilibria, the senders mix between truthful and exaggerated reports. Revelation is a probabilistic phenomenon in the sense that the decision maker fully learns almost every state with some positive probability.

Few other papers study communication with misreporting costs and multiple senders. Among these, \cite{emons2009accuracy} analyze a model with two equally informed senders with opposed interests. The state, report, and action space is the real line. They find a robust fully revealing equilibrium where senders' inflate their reports in opposite directions. As in \cite{kartik2007}, the receiver can invert the senders' reports to fully learn their private information. By contrast, only non-revealing equilibria are robust in my setting where the state space can be unbounded but receiver's action space is discrete.\footnote{Fully revealing outcomes do not survive refinements even in single-sender settings where the receiver's action space is discrete and the state space is continuous and unbounded \citep{vaccari2021influential}.} \cite{kartik2021} consider a model where each sender gets a conditionally independent signal about the state.\footnote{See Section IIID therein. Since senders have different and noisy information, there cannot be equilibria where the state is fully revealed.} They show that senders' strategies are strategic complements: as the misreporting costs of one sender increase, all other senders reveal more information. I obtain an analogous result in my model with perfectly informed senders.\footnote{Such a result follows by performing comparative statics with respect to $k_j$ on equation~\eqref{eq:swing} and looking at its implications on the adversarial equilibrium strategies as in Appendix~\ref{app:aestrat}.}

The combination of multiple senders and signaling costs generates in this paper a framework that resembles those used in all-pay contests.\footnote{This \emph{all-pay} feature is missing in related multi-sender signaling models of electoral competition \citep{banks1990,callander2007}, where only the elected candidate incurs the signaling cost.} Many applications in different areas study persuasion by contending parties as an all-pay contest \citep{skaperdas2012persuasion}. The use of an exogenous ``success function'' eliminates three potential problems. First, it can discard fully revealing outcomes in settings where information asymmetries have relevant consequences (such as in lobbying, campaigning, or litigation). Second, it can circumvent the problem of multiplicity and robustness of equilibria. Third, it can provide tractability.  The analysis of adversarial equilibria shows that all these features are obtained in a setting where the decision maker's behavior is endogenously determined rather than exogenously fixed.

\section{The Model}\label{sec:model}

\noindent {\bf Setup and timeline.} There are three players: two informed senders (1 and 2) and one uninformed decision maker ($\dm$). Let $\theta\in\Theta\subseteq \mathbb{R}$ be the underlying state, distributed according to the full support probability density function $f$. After observing the realized state $\theta$, two senders privately deliver to the decision maker a report $r_j\in \Theta$, where $r_j$ is the report of sender $j$ (he). The decision maker (she), after observing the pair of reports $(r_1,r_2)$ but not the state $\theta$, selects an alternative $a\in\left\{\p,\n\right\}$.

\noindent {\bf Payoffs.} Player $i\in\{1,2,\dm\}$ obtains a payoff of $u_i(a,\theta)$ if the decision maker selects alternative $a$ in state $\theta$. I normalize $u_i\left(\n,\theta\right)=0$ for all $\theta\in\Theta$ and let $u_i(\theta)\equiv u_i(\p,\theta)$, where $u_i(\theta)$ is weakly increasing in $\theta$. The state is a valence or vertical differentiation score, and it is interpreted as the relative quality of alternative $\p$ with respect to alternative $\n$. The decision maker's expected utility from selecting $\p$ given the senders' reports is $U_\dm(r_1,r_2)$.

\noindent {\bf Misreporting costs.} Sender $j$ bears a cost $k_jC_j(r_j,\theta)$ for delivering report $r_j$ when the state is $\theta$. The cost function $C_j$ is continuous, strictly increasing in $|r_j-\theta|$, differentiable for all $r_j\neq\theta$, and satisfies $C_j(\theta,\theta)=0$ for every $\theta\in\Theta$. That is, misreporting is increasingly costly in the magnitude of misrepresentation, while truthful reporting is always costless. Moreover, for every $r_j\in\Theta$, $C_j(r_j,\theta)>C_j(r_j,\theta')$ if $|r_j-\theta|>|r_j-\theta'|$. That is, reports are cheaper when delivered from closer states. The scalar $k_j>0$ is a finite parameter measuring the intensity of misreporting costs.  Sender $j$'s total utility is
\[
w_j(r_j,a,\theta)= \mathds{1}{\{ a=\p \}} u_j(\theta) - k_jC_j(r_j,\theta),
\] 
where $\mathds{1}{\{\cdot\}}$ is the indicator function. It follows that, conditional on the decision maker's eventual choice, both senders prefer to deliver reports that are closer to the truth.

\noindent {\bf Definitions and assumptions.} For simplicity, I assume that the state and report space coincide with the real line, that is, $\Theta=\mathbb{R}$. All results carry through provided that $\Theta$ is large enough (see Appendix~\ref{sec:app}). A generic report $r$ has the literal or exogenous meaning of ``The state is $\theta=r$.'' I say that sender $j$ reports truthfully when $r_j=\theta$, and misreports otherwise. Sometimes, I use $-j$ to denote the sender who is not sender~$j$. 

I define the \emph{threshold} $\tau_i$ as the state in which player $i$ is indifferent between the two alternatives. Formally, $\tau_i:=\{\theta\in\Theta \,|\, u_i(\theta)=0\}$. I assume that utilities $u_i(\theta)$ are such that $\tau_i$ exists and is unique\footnote{These assumptions are for notational convenience. The model can accommodate for senders that always strictly prefer one alternative over the other and for utility functions such that  $u_i(\theta)\neq 0$ for every $\theta\in \Theta$, including step utility functions.} for every $i\in\{1,2,\dm\}$. The threshold $\tau_i$ tells us that player $i$ prefers $\p$ to $\n$ when the state $\theta$ is greater than $\tau_i$. Throughout the paper, I consider the case where senders have opposing biases, i.e., $\tau_1<\tau_\dm<\tau_2$. To make the problem non-trivial I let $\tau_\dm\in\Theta$, and I normalize $\tau_\dm=0$. Under this configuration, the decision maker prefers to select the positive alternative $\p$ when the state $\theta$ takes positive values, and prefers to select the negative alternative $\n$ when the state is negative. I assume that when the decision maker is indifferent between the two alternatives at given beliefs, she selects $\p$.

\begin{figure}
	\centering
	\includegraphics[width=0.65\linewidth]{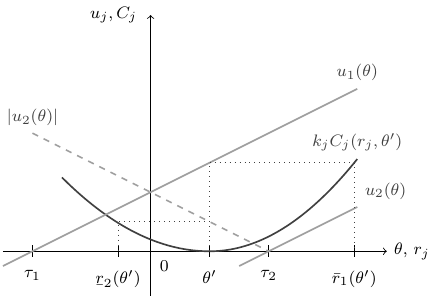}
	\caption{The senders' reaches in state $\theta'$ are marked in the horizontal axis as $\bar r_1(\theta')$ and $\ubar r_2(\theta')$. The misreporting cost function $k_jC_j$, here common to both senders and depicted in full black, is increasing in the magnitude of misreporting. The senders' utilities from the selection of the positive alternative, in full grey, are increasing in the realized state.}
	\label{fig:reach}
\end{figure}

I define the \emph{reach} of sender $j$ in state $\theta$ as the report whose associated misreporting costs offset $j$'s gains from having his own preferred alternative eventually selected. In other words, reports more expensive than the reach are strictly dominated by truthful reporting. As a result, in every equilibrium senders deliver only reports that are within their own reach. Formally, the reach of sender~$1$ in state $\theta$ is defined as\footnote{The definition of reach is sender-dependent, as senders play monotonic strategies in the equilibria we consider. Specifically, sender~1 inflates the realized state, whereas sender~2 belittles it (see Lemma~\ref{lemma:monot}).}
\begin{equation}\label{eq:reach1}
\bar r_1(\theta) := \max\left\{r\in\mathbb{R} \; \text{s.t.} \;  |u_1\left(\theta\right)|= k_1C_1\left(r,\theta\right) \right\}. 
\end{equation}
Similarly, the reach of sender~$2$ in state $\theta$ is defined as
\begin{equation}\label{eq:reach2}
\ubar r_2(\theta) := \min \left\{r\in\mathbb{R} \; \text{s.t.} \;   |u_2\left(\theta\right)|= k_2C_2\left(r,\theta\right) \right\}. 
\end{equation}
I will sometimes use the \emph{inverse reaches} $\bar r_1^{-1}(r_1)$ and $\ubar r_2^{-1}(r_2)$, where $\bar r_1^{-1}(\cdot)$ and $\ubar r_2^{-1}(\cdot)$ map from $\Theta$ to $\Theta$, and are defined as the inverse functions of $\bar r_1(\theta)$ and $\ubar r_2(\theta)$, respectively. Figure~\ref{fig:reach} illustrates the senders' reaches in a state $\theta'\in(\tau_1,\tau_2)$.

\noindent {\bf Strategies.} A pure strategy for sender $j$ is a function $\rho_j:\Theta \to \Theta$ such that $\rho_j(\theta)$ is the report delivered by sender $j$ in state $\theta$. A mixed strategy for sender $j$ is a mixed probability measure $\phi_j: \Theta \to \Delta(\Theta)$, where $\phi_j(r_j,\theta)$ is the mixed probability density that $\phi_j(\theta)$ assigns to a report $r_j\in \Theta$. The cumulative distribution function (CDF) of $\phi_j$ is $\Phi_j$. I denote by $S_j(\theta)$ the support of sender $j$'s strategy in state $\theta$. Appendix~\ref{app:notation} introduces additional notation that is required to study equilibria in mixed strategies.

I say that a pair of reports $(r_1,r_2)$ is off path if, given the senders' strategies, $(r_1,r_2)$ will never be observed by the decision maker. Otherwise, I say that the pair $(r_1,r_2)$ is on path. A posterior belief function for the decision maker is a mapping $p:\Theta^2\to\Delta(\Theta)$ that, given any pair of reports $(r_1,r_2)$, generates posterior beliefs $p(\theta \,|\, r_1,r_2)$ with CDF $P(\theta\,|\,r_1,r_2)$. Given a pair of reports $(r_1,r_2)$ and posterior beliefs $p(\theta \,|\, r_1,r_2)$, the decision maker selects an alternative in the sequentially rational set $\beta(r_1,r_2)$, where
\[
\beta(r_1,r_2)=\argmax_{a\in\{\text{\scriptsize \textcircled{+}},\text{ \scriptsize \textcircled{-}}\}}\mathbb{E}_p\left[u_\dm (a,\theta) \,|\, r_1,r_2\right].
\]
As mentioned above, if $p(\theta\,|\,r_1,r_2)$ is such that $U_\dm(r_1,r_2)=0$, then $\beta(r_1,r_2)=\p$.

\subsection{Solution Concepts}\label{sec:solutions}

The solution concept is perfect Bayesian equilibrium (PBE).\footnote{For a textbook definition of perfect Bayesian equilibrium, see \cite{fudenberg1991game}.} To avoid outcomes that are sustained by an excessively arbitrary interpretation of the senders' reports, I restrict attention to equilibria where posterior beliefs $p$ satisfy the following weak monotonicity condition: for every $r_j\geq r_j'$ and $j\in\{1,2\}$,
\begin{displaymath}\tag{wM}\label{eq:m}
U_\dm(r_1,r_2)\geq U_\dm(r_1',r_2').
\end{displaymath}
The \eqref{eq:m} condition\footnote{Posterior beliefs $p(\theta\,|\,r_1,r_2)$ first-order stochastically dominate $p(\theta\,|\,r_1',r_2')$ for $r_j\geq r_j'$, $j\in\{1,2\}$, if and only if $\int u(\theta)p(\theta\,|\,r_1,r_2)d\theta\geq\int u(\theta)p(\theta\,|\,r_1',r_2')d\theta$ for every weakly increasing utility function $u(\theta)$. Thus, \eqref{eq:m} is weaker than first-order stochastic dominance as it applies only to $u(\theta)\equiv u_\dm(\theta)$.} says that a higher report cannot signal to the decision maker a lower expected utility from selecting alternative $\p$. Focusing on these equilibria is natural given that the value of $\p$ is increasing in the state, reports are literal, and misreporting is costly. An immediate implication of this condition is that senders' strategies are monotonic: in every equilibrium satisfying \eqref{eq:m}, each sender either reports truthfully or exaggerates the relative quality of his preferred alternative.

\begin{restatable}{lemma}{monotlemma}\label{lemma:monot}
	In every perfect Bayesian equilibrium satisfying \eqref{eq:m} we have that $r_j
	\geq \theta$ for all $\theta\geq \tau_j$, and $r_j
	\leq \theta$ for all $\theta\leq \tau_j$, $j\in\{1,2\}$.
\end{restatable}

Hereafter, I refer to perfect Bayesian equilibria of the game described in this section that satisfy condition~\eqref{eq:m} simply as equilibria. As we shall see in Section~\ref{sec:fre}, condition~\eqref{eq:m} is not sufficient to rule out equilibria with undesirable properties. To analyze sensible outcomes, I further study a class of equilibria that satisfies two additional restrictions on the decision maker's posterior beliefs. I refer to equilibria satisfying these additional conditions as adversarial equilibria, and analyze them in Section~\ref{sec:mixed}.

\begin{definition}\label{def:directeqa}
	An adversarial equilibrium (AE) is a perfect Bayesian equilibrium of the game described in Section~\ref{sec:model} where the decision maker's posterior beliefs $p$ satisfy the following conditions:
	\begin{enumerate}
		\item[i)] The \eqref{eq:m} condition holds, and for every pair of reports $(r_1,r_2)$ such that $\ubar r_2(0)<r_2\leq 0 \leq r_1<\bar r_1(0)$, and for $j\in\{1,2\}$, we have
		\begin{equation}
			\frac{d U_\dm(r_1,r_2)}{d r_j}>0; \tag{sM}\label{eq:d}
		\end{equation}
		\item[ii)] Once the decision maker observes the pairs of reports $\left(\bar r_1(0),\ubar r_2(0)\right)$ and $(0,0)$, her posterior beliefs $p$ are such that she is indifferent between the two alternatives, i.e., 
		\begin{equation}
			U_\dm\left(\bar r_1(0),\ubar r_2(0)\right)=U_\dm\left(0,0\right)=0. \tag{Dom}\label{eq:c}
		\end{equation}
	\end{enumerate}
\end{definition}

The first condition, \eqref{eq:d}, imposes a stronger monotonicity restriction on posterior beliefs $p$ than \eqref{eq:m}, but only for pairs of reports consisting of conflicting recommendations. Otherwise, \eqref{eq:m} applies. Since \eqref{eq:d} implies \eqref{eq:m}, Lemma~\ref{lemma:monot} applies also to adversarial equilibria. Intuitively, \eqref{eq:d} means that strictly higher conflicting reports inform the decision maker that the expected value of selecting alternative $\p$ is strictly higher. 

Condition~\eqref{eq:c} draws on a simple dominance argument. Recall that, by definition of reach, sender~1 prefers to tell the truth than to deliver $\bar r_1(0)$ when the realized state is strictly negative. Likewise, sender~2 prefers to tell the truth than to deliver $\ubar r_2(0)$ when the realized state is strictly positive. Upon observing the pair of reports $(\bar r_1(0),\ubar r_2(0))$, the decision maker should conjecture that the realized state is zero: for otherwise, it must be that one of the two senders is delivering a strictly dominated report. A similar logic applies to the pair of reports $(0,0)$. Recall that the senders' equilibrium strategies are monotonic (Lemma~\ref{lemma:monot}). Upon observing $(0,0)$, the decision maker should conjecture that the realized state is zero, for otherwise one of the two senders must be delivering a report that is strictly dominated.\footnote{Condition \eqref{eq:c} does not require that the decision maker's posterior beliefs be degenerate at $0$. As we shall see, in every adversarial equilibrium the pair $(\bar r_1(0),\ubar r_2(0))$ is on path only for $\theta=0$, and thus it fully reveals that the state is indeed zero. By contrast, no sender ever delivers $r_j=0$ on path, and thus the pair of reports $(0,0)$ is not only off path but it must constitute a double deviation.}

\section{Receiver-efficient Equilibria and Robustness}\label{sec:fre}

The goal of this section is to study how costly talk communication affects the existence and properties of two important classes of equilibria: babbling and fully revealing. No information is transmitted in the former, whereas the decision maker always learns the state in the latter. Typically, in cheap talk games there is a babbling equilibrium, while in standard signaling and disclosure games there is a fully revealing equilibrium (FRE). In the setup considered here, full revelation can be naturally achieved when senders play truthful strategies, that is, when they always report truthfully the realized state. As the next lemma shows, the introduction of misreporting costs prevents the existence of both babbling and truthful equilibria. 
\begin{restatable}{lemma}{lemmababbling}\label{lemma:babbling}
	There are no babbling equilibria. Misreporting occurs in every equilibrium.
\end{restatable}

Intuitively, babbling cannot occur because, when misreporting is costly, ignored senders best respond by reporting truthfully. Truthful equilibria do not exist because there are always situations where senders can profit from lying if their competitor reports truthfully.\footnote{Lemma~\ref{lemma:babbling} applies to all perfect Bayesian equilibria of the game described in Section~\ref{sec:model}, and not only to those satisfying~\eqref{eq:m}. \cite{battaglini2002} uses an argument similar to the revelation principle to show that, in a multi-sender cheap talk model, if there exists a fully revealing equilibrium then there exists a fully revealing equilibrium in truthful strategies. Such an argument cannot be applied here because of the presence of misreporting costs. In this model there are fully revealing equilibria, but there are no equilibria in truthful strategies.} Since fully revealing outcomes do not necessarily require senders to play truthful strategies, Lemma~\ref{lemma:babbling} does not rule out the existence of fully revealing equilibria.

In what follows, I show that there are equilibria where the decision maker gets the full information payoff.\footnote{Fully revealing equilibria of this model naturally exist when one of the two senders is unbiased or cannot misreport. The former occurs when $\tau_j=\tau_{dm}$, while the latter occurs when $k_j=\infty$.} In the setting studied here, the combination of a rich state space together with a binary action space  implies that the decision maker does not need to know the realized state in order to select her favorite alternative. All she needs to know is whether the state is positive or negative. For the purposes of this section, studying fully revealing equilibria would be too restrictive. The following definition gives a weaker notion of revelation that will prove useful for the analysis that follows.

\begin{definition}
	A receiver-efficient equilibrium (REE) is an equilibrium where for every $\theta\in\Theta$, $r_j\in S_j(\theta)$, and $j\in\{1,2\}$, we have $\beta(r_1,r_2)=\p$ if $\theta\geq 0$, and $\beta(r_1,r_2)=\n$ otherwise.
\end{definition}

\begin{figure}
	\centering
	\includegraphics[width=0.5\linewidth]{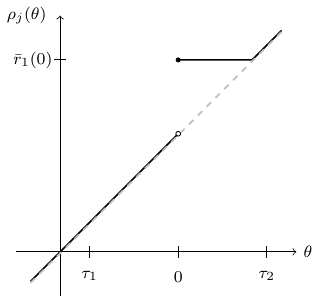}
	\caption{Senders' strategies in a receiver-efficient and fully revealing equilibrium. The reporting rules of senders~1 and 2 are indicated by black and dashed gray lines, respectively.}
	\label{fig:fre}
\end{figure}

A fully revealing equilibrium is also receiver-efficient, but a receiver-efficient equilibrium is not necessarily fully revealing. Figure~\ref{fig:fre} shows reporting strategies that not only constitute a receiver-efficient equilibrium, but are also fully revealing. To verify that Figure~\ref{fig:fre} depicts an equilibrium, consider the following strategies: sender~1 delivers $\rho_1(\theta)=\bar r_1(0)$ for every $\theta\in[0,\bar r_1(0)]$, where for simplicity we assume that $\bar r_1(0)<\tau_2$. Otherwise, sender~1 reports truthfully. By contrast, sender~2 always reports truthfully, i.e., $\rho_2(\theta)=\theta$ for all $\theta\in\Theta$. Given any on path pair of reports, posterior beliefs are such that $P(\theta\,|\,r_1,r_2)=0$ for every $\theta<r_2$ and $P(\theta\,|\,r_1,r_2)=1$ otherwise, which is consistent with sender~2 playing a separating strategy. Off path beliefs are such that  $U_\dm(r_1,r_2)<0$ if $r_1<\bar r_1(0)$, and $P(\theta\,|\,r_1,r_2)=1$ if and only if $\theta\geq r_1 \geq \bar r_1(0)$. By the definition of reach, sender~1 would never find it profitable to deliver a report $r_1\geq \bar r_1(0)$ when $\theta<0$. Sender~2 cannot deviate from his truthful strategy by delivering a negative report when the state is positive: since $\rho_1(\theta)\geq \bar r_1(0)$ for every $\theta\geq 0$, such a deviation would induce $\beta(\cdot,\cdot)=\p$. No sender has a profitable individual deviation from the prescribed equilibrium strategies. As a result, there exist equilibria where senders always fully reveal the state to the decision maker, even though full revelation involves misreporting.

Figure~\ref{fig:fre} also illustrates the existence of equilibria in pure strategies. Intuitively, when senders play pure strategies and in each state at least one of the two senders plays a separating strategy, then the decision maker can always \emph{invert} their reports to recover the underlying truth. This argument suggests that all pure-strategy equilibria are receiver-efficient. The next lemma shows that such an intuition is correct and, in addition, that all receiver-efficient equilibria are in pure strategies. 

\begin{restatable}{lemma}{pureREE}\label{lemma:pureREE}
	An equilibrium is receiver-efficient if and only if it is in pure strategies.
\end{restatable}

The receiver-efficient equilibrium discussed above is, however, problematic. To see what the problem is, consider again the strategies depicted in Figure~\ref{fig:fre} and a state $\theta'\in(0,\bar r_1(0))$. Suppose that in state $\theta'$ sender~1 deviates from the prescribed equilibrium by reporting the truth instead of $\rho_1(\theta')=\bar r_1(0)$, whereas sender~2 sticks to his separating reporting rule. Notice that, in the equilibrium under consideration, sender~1 never delivers $r_1=\theta'$. Once the decision maker receives the off path pair of reports $(\theta',\theta')$, her posterior beliefs $p$ induce an expected payoff of $U_\dm(\theta',\theta')<0$,
leading to $\beta(\theta',\theta')=\n$. These off path beliefs require the decision maker to conjecture that the state is likely to be negative. However, this means that the decision maker must entertain the possibility that (i) both senders simultaneously deviated from the prescribed equilibrium strategies, and (ii) sender~2 purposefully delivered a strictly dominated report.

In the remaining part of this section, I test receiver-efficient equilibria using two well-known refinements that apply to games with multiple senders: unprejudiced beliefs \citep{bagwell1991} and $\varepsilon$-robustness \citep{battaglini2002}.

\noindent{\bf Unprejudiced beliefs.} Consider again a deviation from the equilibrium depicted in Figure~\ref{fig:fre} where both senders report truthfully in some state $\theta'\in(0,\bar r_1(0))$. If, whenever possible, the decision maker conjectures deviations as individual and thus as originating from one sender only, then she should infer that sender~1  performed the deviation: sender~1 never reports $r_1=\theta'$ on the equilibrium path, whereas sender~2 truthfully reports $r_2=\theta'$ only when the state is indeed $\theta'$. Since sender~2 is following his separating strategy, the decision maker should infer that the state is $\theta'>0$. According to this line of reasoning, off path beliefs must be such that $P(\theta\,|\,\theta',\theta')=1$ if and only if $\theta\geq\theta'$, and thus $\beta(\theta',\theta')=\p$. Such a deviation becomes profitable for sender~1 because it economizes on misreporting costs without affecting the outcome.

\cite{bagwell1991} introduce the concept of ``unprejudiced beliefs,'' which formalize the idea that the decision maker should rule out the possibility that multiple senders are deviating at the same time whenever it is possible that only a single sender is deviating. \cite{vida2021} show that, in generic multi-sender signaling games, strategic stability \citep{kohlberg1986} implies unprejudiced beliefs. Apart from their association with the notion of strategic stability, unprejudiced beliefs are intuitive, easily applicable, and consistent with the notion of Nash equilibrium and, as such, constitute a sensible way to refine equilibria in multi-sender signaling games when other criteria fail to do so. The following definition formalizes unprejudiced beliefs.\footnote{Definition~\ref{def:unprejudiced} is weaker than the definition originally introduced by \cite{bagwell1991}.}

\begin{definition}[\citealp{vida2021}]\label{def:unprejudiced}
	Given senders' strategies $\rho_j$, the decision maker's posterior beliefs $p$ are unprejudiced if, for every pair of reports $(r_1,r_2)$ such that $\rho_j(\theta')=r_j$ for some $\theta'\in\Theta$ and $j\in\{1,2\}$, we have that $p(\theta'' \,|\, r_1,r_2)>0$ only if there is a sender $i\in\{1,2\}$ such that $\rho_i(\theta'')=r_i$.
\end{definition}

We have seen how the above informational free-riding argument breaks down the receiver-efficient equilibrium depicted in Figure~\ref{fig:fre}. A natural question is whether such an argument applies only in that particular case or if instead it rules out other equilibria. The next proposition tells us that in fact there is no receiver-efficient equilibrium that supports unprejudiced beliefs.

\begin{restatable}{proposition}{nonREE}\label{prop:nonREE}
	There are no receiver-efficient equilibria with unprejudiced beliefs.
\end{restatable}

\noindent {\bf $\varepsilon$-robustness.} In the model described in Section~\ref{sec:model}, senders are perfectly informed and the decision maker can perfectly observe the senders' reports. In other words, there are no perturbations, or ``noise,'' in the senders' report or the decision maker's observations. This modeling strategy allows me to isolate the effects of strategic interactions from the effects of statistical information aggregation. At the same time, however, it allows for excessive freedom to pick ad-hoc beliefs that would not survive the presence of even arbitrarily small perturbations in the transmission of information.

I follow \cite{battaglini2002} and define an $\varepsilon$-perturbed game
as the game described in Section~\ref{sec:model}, in which the decision maker perfectly observes the report of sender $j$ with probability $1-\varepsilon_j$ and with probability $\varepsilon_j$ observes a random report $\tilde r_j$, where $\tilde r_j$ is a random variable with continuous distribution $G_j$, density $g_j$, and support in $\Theta$. This may correspond to a situation where with some probability the decision maker misreads reports; or, alternatively, where with some probability senders deliver a wrong report by mistake.\footnote{\cite{battaglini2002} perturbs the senders' observation of the realized state, whereas I perturb the decision maker's observed reports. My perturbation is qualitatively similar to Battaglini's.} As before, senders incur misreporting costs that depend only on the realized state $\theta$ and on their ``intended'' report $r_j$, and not on the wrongly observed or delivered $\tilde r_j$. The introduction of noise makes any pair of reports possible on the equilibrium path. The decision maker's posterior beliefs depend on $\varepsilon=(\varepsilon_1,\varepsilon_2)$, $G=(G_1,G_2)$, and the senders' reporting strategies $\rho_j(\theta)$.

\begin{definition}[\citealp{battaglini2002}]
	An equilibrium is $\varepsilon$-robust if there exist a pair of distributions $G=(G_1,G_2)$ and a sequence $\varepsilon^n=(\varepsilon^n_1,\varepsilon^n_2)$ converging to zero such that the off path beliefs of the equilibrium are the limit of the beliefs that the equilibrium strategies would induce in an $\varepsilon$-perturbed game as $\varepsilon^n\to 0^+$.
\end{definition}

Intuitively, as the noise $\varepsilon$ fades away, the event in which the decision maker misreads both reports becomes negligible. At the limit as $\varepsilon \to 0^+$, the decision maker infers that she is correctly observing at least one of the two reports. Once she observes an off path pair of reports in an $\varepsilon$-robust equilibrium, the decision maker conjectures---whenever possible---that one sender is following his prescribed reporting strategy while the other is not. This last implication of $\varepsilon$-robustness suggests that there might be a tight connection between the refinement criteria of $\varepsilon$-robustness and unprejudiced beliefs. The next lemma confirms the existence of such a relationship.\footnote{Lemma~\ref{lemma:epsilonunprejudiced} applies to perfect Bayesian equilibria of the game described in Section~\ref{sec:model} with $n\geq 2$ senders.}

\begin{restatable}{lemma}{epsilonrobust}\label{lemma:epsilonunprejudiced}
	If an equilibrium is $\varepsilon$-robust, then it has unprejudiced beliefs.
\end{restatable}

A straightforward implication of Lemma~\ref{lemma:epsilonunprejudiced} and Proposition~\ref{prop:nonREE} is that no receiver-efficient or fully revealing equilibrium is $\varepsilon$-robust. By Lemma~\ref{lemma:pureREE}, we obtain that also pure-strategy equilibria are not supported by unprejudiced beliefs and are not $\varepsilon$-robust. These results suggest that mixed-strategy equilibria are qualitatively important. The next section is dedicated to finding equilibria that are robust in the sense that they are $\varepsilon$-robust.

\section{Adversarial Equilibria}\label{sec:mixed}

Findings in the previous section show that pure-strategy and receiver-efficient equilibria exist, but are supported by off path beliefs with potentially undesirable characteristics. Such results motivate the quest for robust equilibria in mixed strategies. Since condition \eqref{eq:m} does not rule out receiver-efficient outcomes, a different set of restrictions is required to obtain robust equilibria. This section aims to identify sufficient conditions under which equilibria are robust, characterize such robust equilibria, and show that they have desirable properties. All proofs and a number of intermediate results are relegated to Appendix~\ref{app:mixed}.

I focus the following analysis on adversarial equilibria as described in Section~\ref{sec:solutions}. For an immediate application of this solution concept, let's revisit the receiver-efficient equilibrium depicted in Figure~\ref{fig:fre}. To prevent a deviation by sender~1, the decision maker's posterior beliefs $p$ must be such that she selects $\n$ when sender~1's report is lower than $\bar r_1(0)$. However, such beliefs cannot be part of an adversarial equilibrium. By \eqref{eq:c}, the decision maker is indifferent between the two alternatives when both senders claim that the state is zero. As a result, by \eqref{eq:d} the decision maker expects the state to be positive when both senders claim that the state is positive, and accordingly selects $\p$. The language structure imposed by conditions \eqref{eq:d} and \eqref{eq:c} is sufficient to rule out receiver-efficient equilibria such as the one depicted in Figure~\ref{fig:fre}. The next lemma confirms that this pruning power extends in general to all receiver-efficient equilibria.
\begin{restatable}{lemma}{reenotde}\label{lemma:reenotde}
	Receiver-efficient equilibria are not adversarial equilibria.
\end{restatable}

Conditions underpinning adversarial equilibria rule out all receiver-efficient and, by Lemma~\ref{lemma:pureREE}, all pure-strategy equilibria. However, introducing conditions on the decision maker's beliefs may raise some concerns. First, an equilibrium satisfying both \eqref{eq:d} and \eqref{eq:c} may not exist.\footnote{Even well-behaved signaling games may have no equilibria \citep{manelli1996}.} Second, provided that adversarial equilibria exist, there may be an issue of multiplicity. Lastly, even adversarial equilibria may not be robust according to the criteria defined in the previous section. The following result shows that these issues are not present in the current setting.

\begin{theorem}\label{th:advthm}
	Adversarial equilibria of the game described in Section~\ref{sec:model},
	\begin{itemize}[noitemsep,topsep=0pt]
		\item[i)] always exists;
		\item[ii)] are essentially unique in terms of equilibrium outcomes and strategies;
		\item[iii)] are $\varepsilon$-robust, or outcome- and strategy-equivalent to AE that are $\varepsilon$-robust.\footnote{An adversarial equilibrium may not be $\varepsilon$-robust because of its off path beliefs. Part ii) of Theorem~\ref{th:advthm} states that there can be multiple AE sharing the same strategies and outcomes, but differing in their off path beliefs. Part iii) tells us that there always exists an AE that is $\varepsilon$-robust and thus unprejudiced.}
	\end{itemize} 
\end{theorem}

Next, I provide some intuition and definitions that will prove useful to understand the players' strategies in an adversarial equilibrium. To fix ideas, suppose from now on that the realized state is positive, $\theta>0$. Since senders play monotonic strategies (Lemma~\ref{lemma:monot}), sender~1 delivers a positive report, say $r_1\geq\theta$. As we have seen in the previous example, the interaction of conditions~\eqref{eq:d} and \eqref{eq:c} has an important consequence: when both senders' reports are positive, the decision maker correctly infers that the state is positive as well. As a result, sender~2 must deliver a negative report to achieve persuasion. 

It follows that persuasion can take place only when senders deliver reports with conflicting signs. By condition~\eqref{eq:d}, we know that in these cases the decision maker expects strictly lower reports to originate from strictly lower states. To achieve persuasion, sender~2 must deliver a negative report that is sufficiently low. Say that, given sender~1's report, the decision maker selects $\n$ when sender~2 claims that the state is lower than $s(r_1)$, and selects $\p$ otherwise. We shall say that $s(r_1)$ \emph{swings} the decision maker's choice. The notion of \emph{swing report} is key for understanding adversarial equilibria, and the following definition formalizes such a concept.

\begin{definition}\label{def:swing}
	Given a report $r$, the swing report $s(r)$ is defined as
	\[
	s(r)=
	\begin{cases}
		\left\{r_2\in R_{2} \,|\, U_\dm(r,r_2)=0\right\} & \text{if sender~1's report is } r\geq 0 \\
		\left\{r_1\in R_{1} \,|\, U_\dm(r_1,r)=0\right\} & \text{if sender~2's report is } r < 0
	\end{cases}
	\]
	If $s(r)=\varnothing$, then I set $s(r)=-\infty$ when $r\geq 0$ is delivered by sender~1, and $s(r)=\infty$ when $r < 0$ is delivered by sender~2.
\end{definition}

With a slight abuse of language, I hereafter say that sender~$j$ swings the report of his opponent $-j$ whenever the pair of reports $(r_1,r_2)$ induce the selection of sender~$j$'s preferred alternative. Sender~1 swings the report of sender~2 when $r_1\geq s(r_2)$. Similarly, sender~2 swings the report of sender~1 when $r_2<s(r_1)$.

In adversarial equilibria, the swing report $s(r)$ has a number of intuitive properties: first, condition~\eqref{eq:d} ensures that the swing report, if it exists, is unique; second, condition~\eqref{eq:c} pins down the swing report for $s(\bar r_1(0))=\ubar r_2(0)$, $s(\ubar r_2(0))=\bar r_1(0)$, and $s(0)=0$. From the interaction of \eqref{eq:c} and \eqref{eq:d}, it follows that every report $r$ has a unique swing report $s(r)$ such that if $r>0$ then $s(r)<0$ (resp. $r<0$ has $s(r)>0$). The swing report of a swing report is the report itself, that is, $s(s(r))=r$. Lastly, higher reports have lower swing reports.\footnote{See Lemma~\ref{lemma:swing} in Appendix~\ref{app:mixed}.} Since $s(r)$ is a strictly decreasing function of $r$, I shall refer to $s(r)$ as the \emph{swing report function}.

When the state takes extreme values, a sender may not be able to swing the report of his opponent profitably. In particular, persuasion is always prohibitively expensive when $s(\theta)$ is beyond a sender's reach. For example, sender~2 cannot profitably achieve persuasion when the state is such that $s(\theta)<\ubar r_2(\theta)$, as $r_1\geq\theta$ and $s(r_1)$ is decreasing in $r_1$. In such cases, we should expect both senders to report truthfully and deliver matching reports that reveal the state. It is helpful to define cutoffs in the state space that help determine when truthful reporting always occurs in adversarial equilibria. I refer to these cutoffs as the \emph{truthful cutoffs}, and define them as follows.

\begin{definition}\label{def:cutoffs} The truthful cutoffs $\theta_l$ and $\theta_h$ are defined as
	\begin{displaymath}
	\theta_l := \left\{\theta\in\Theta \,|\, s(\theta) = \bar r_1(\theta) \right\},
	\end{displaymath}
	\begin{displaymath}
	\theta_h := \left\{\theta\in\Theta \,|\, s(\theta) = \ubar r_2(\theta) \right\}.
	\end{displaymath}
\end{definition}

The above cutoffs are states in which senders are at best indifferent between reporting truthfully and achieving persuasion by misreporting. Because of the properties of $s(\cdot)$, the cutoffs are such that $\theta_l<0<\theta_h$. Moreover, senders have a conflict of interest in every state within the truthful cutoffs.\footnote{That is, $\tau_1\leq\theta_l<\theta_h\leq\tau_2$. See Lemma~\ref{lemma:cutoffs} in Appendix~\ref{app:mixed}.} As mentioned before, we should expect senders to always report truthfully---and therefore to play pure strategies---when the state lies outside the set $(\theta_l,\theta_h)$. The interaction between senders becomes more intricate when the state takes values within the truthful cutoffs.

Consider a situation where the state is positive and within the truthful cutoffs, that is, $\theta\in(0,\theta_h)$. Suppose that sender~2 expects sender~1 to report truthfully. To convince the decision maker to select $\n$, sender~2 would deliver a negative report that is both effective and affordable, i.e., $r_2\in(\ubar r_2(\theta),s(\theta))$. In equilibrium, sender~1 would anticipate sender~2's report $r_2$. As a result, sender~1 would deliver $r_1=s(r_2)>\theta$ to neutralize sender~2's persuasion attempt. However, sender~2 can foresee this reaction, and expects sender~1 to deliver $r_1$ instead of reporting truthfully. Consequently, sender~2 would deliver $r_2'\in(\ubar r_2(\theta),s(r_1))$, where $r_2'<r_2$. Senders are now bearing higher misreporting costs in the attempt to defeat their opponent. Since competing forces escalate expenditures, sender~2 would eventually expect its opponent to deliver a report he cannot afford to swing, that is, a $r_1'$ such that $s(r_1')<\ubar r_2(\theta)$. In this case, sender~2 would report truthfully to economize on costs. Anticipating this reaction, sender~1 would follow suit, taking us back to our starting point where sender~1 reports truthfully the state.\footnote{Recall that, as mentioned at the beginning of Section~\ref{sec:mixed}, in adversarial equilibria the decision maker selects the positive alternative after observing two positive reports.} This informal example suggests that senders play mixed strategies in states within the truthful cutoffs.

Conditions~\eqref{eq:d} and \eqref{eq:c} ensure that the decision maker's problem is trivial when reports have the same sign. In these cases, the state is fully revealed by the sender recommending his least preferred alternative. By contrast, the problem of strategic inference is more complicated when reports have conflicting signs. In these cases, the decision maker cannot determine whether the state is positive and sender~2 is misreporting or whether the state is negative and sender~1 is misreporting. In the attempt to select the correct alternative, the decision maker compares and cross-validates the reports. When doing so, she assigns a weight to the senders' reports that depends on their relative characteristics. This process determines the swing report function discussed above.

The next proposition characterizes the senders' reporting strategies and the decision maker's beliefs in adversarial equilibria. Recall that $\Phi_j:\Theta^2\to[0,1]$ is a cumulative distribution function describing the strategy of sender~$j$. Specifically, $\Phi_j(r_j,\theta)$ denotes the probability that sender~$j$ delivers in state $\theta$ a report that is lower than or equal to $r_j$. Figure~\ref{fig:directeqm} depicts the senders' strategies in an adversarial equilibrium of a setting where players have symmetric features.\footnote{The figure depicts the case where $\Theta=\mathbb{R}$, $f$ is symmetric around zero, $u_1(\theta)=-u_2(\theta)=1$, $k_jC_j(r_j,\theta)=|r_j-\theta|$, and $u_{dm}(\theta)=1$ for $\theta\geq 0$ and $-1$ otherwise. Given these parameters, the reaches are $\bar r_1(\theta)=\theta+1$ and $\ubar r_2(\theta)=\theta-1$. The state is $\theta'=1/4$, the truthful cutoffs are $\theta_h=-\theta_l=1/2$, and the swing report function is $s(r)=-r$. Each misreport $r_j\neq\theta$ is delivered with partial density probability $1/[2(1-2\theta)]$, and each sender reports truthfully with probability $2|\theta'|=1/2$. See also the discussion of symmetric environments in Section~\ref{sec:example}.}

\begin{proposition}\label{prop:adveqm}
	An adversarial equilibrium is a tuple $(\Phi_1,\Phi_2,p)$ such that,
	\begin{itemize}[noitemsep,topsep=0pt]
		\item[i)] If $\theta\notin(\theta_l,\theta_h)$, then $\Phi_j(r_j,\theta)=0$ for $r_j<\theta$, and $\Phi_j(r_j,\theta)=1$ otherwise, $j\in\{1,2\}$. That is, both senders always report truthfully when the realized state lies outside the truthful cutoffs;
		\item[ii)] If $\theta\in(\theta_l,0)$, then the senders' reporting strategies are,
		\begin{displaymath}
			\Phi_1(r_1,\theta)=
			\begin{cases}
				0 & \quad \text{if } r_1<\theta\\ 
				1-\frac{k_2}{-u_2(\theta)} C_2(s(\bar r_1(\theta)),\theta) & \quad \text{if } r_1\in[\theta,s(\theta))\\
				1-\frac{k_2}{-u_2(\theta)} \bigg[C_2(s(\bar r_1(\theta)),\theta)-C_2(s(r_1),\theta) \bigg] & \quad \text{if } r_1\in[s(\theta),\bar r_1(\theta))\\
				1 & \quad \text{if } r_1 \geq  \bar r_1(\theta)
			\end{cases}
		\end{displaymath}
		\begin{displaymath}
			\Phi_2(r_2,\theta)=
			\begin{cases}
				0 & \quad \text{if }
				r_2< s(\bar r_1(\theta))\\ 
				1-\frac{k_1}{u_1(\theta)}C_1(s(r_2),\theta) & \quad \text{if } r_2\in[s(\bar r_1(\theta)),\theta)\\
				1 & \quad \text{if } r_2 \geq \theta
			\end{cases}
		\end{displaymath}
	\item[iii)] If $\theta\in[0,\theta_h)$, then the senders' reporting strategies are,
	\begin{displaymath}
		\Phi_1(r_1,\theta)=
		\begin{cases}
			0 & \quad \text{if } r_1<\theta\\ 
			\frac{k_2}{-u_2(\theta)}C_2(s(r_1),\theta) & \quad \text{if } r_1\in[\theta,s(\ubar r_2(\theta)))\\
			1 & \quad \text{if } r_1 \geq s(\ubar r_2(\theta))
		\end{cases}
	\end{displaymath}
	\begin{displaymath}
		\Phi_2(r_2,\theta)=
		\begin{cases}
			0 & \quad \text{if } r_2< \ubar r_2(\theta)\\ 
			\frac{k_1}{u_1(\theta)} \bigg[C_1(s(\ubar r_2(\theta)),\theta)-C_1(s(r_2),\theta) \bigg] & \quad \text{if } r_2\in[\ubar r_2(\theta),s(\theta))\\
			\frac{k_1}{u_1(\theta)} C_1(s(\ubar r_2(\theta)),\theta)  & \quad \text{if } r_2\in[s(\theta),\theta)\\
			1 & \quad \text{if } r_2 \geq \theta
		\end{cases}
	\end{displaymath}
		\item[iv)] Posterior beliefs $p$ satisfy $\eqref{eq:c}$, $\eqref{eq:d}$, and are such that the swing report function $s(r_i)$ is implicitly defined for $i,j\in\{1,2\}$, $i\neq j$, and $r_i\in \left[\ubar r_2(0),\bar r_1(0)\right]$, as
		\begin{displaymath}
			s(r_i)=\left\{r_j\in \Theta \; \bigg| \; \int_{\max\left\{r_2,\bar r_1^{-1}(r_1)\right\}}^{\min\left\{r_1,\ubar r_2^{-1}(r_2)\right\}} f(\theta)\frac{u_\dm(\theta)}{u_1(\theta) u_2(\theta)}\frac{d C_j(r_j,\theta)}{d r_j}\frac{d C_i(r_i,\theta)}{d r_i}d\theta=0\right\}.
		\end{displaymath}
	\end{itemize}
\end{proposition}


\begin{figure}
	\centering
	\includegraphics[width=0.95\linewidth]{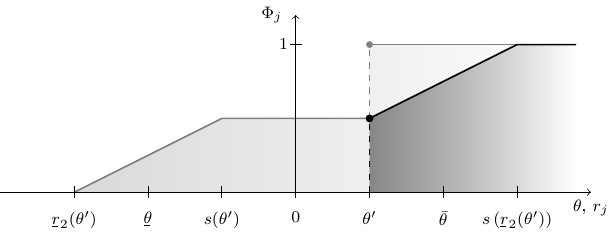}
	\caption{The senders' reporting strategies in an adversarial equilibrium of a setting where misreporting costs are linear and players have symmetric features (see Section~\ref{sec:symmetric}). The CDFs representing sender~1 and 2's reporting strategies when the realized state is $\theta'$ are depicted in black and grey, respectively. The strategies are discontinuous at $\theta'$ because both senders report truthfully with positive probability.} 
	\label{fig:directeqm}
\end{figure}


The proposition shows a discontinuity in the senders' reporting strategies $\Phi_j$ when reports are truthful, that is, at $r_j=\theta$. This discontinuity reflects that senders report truthfully almost every state with positive probability.\footnote{The only exception is $\theta=0$, where the truth is never reported on the equilibrium path.} As noted before, the decision maker fully learns the state upon observing reports that are identical or have the same sign. The state is more likely to be revealed in relatively extreme states, where the decision maker obtains a substantially different payoff from the alternatives.\footnote{See Corollary~\ref{cor:prob} in Appendix~\ref{app:aestrat}.} By contrast, the decision maker is more likely to make mistakes in central states, and closer to the point where her preferred alternative changes. Figure~\ref{fig:probrevelation} depicts the probability of these events, including the likelihood that the decision maker eventually selects the alternative she would choose under perfect information.

\begin{figure}
	\centering
	\includegraphics[width=0.95\linewidth]{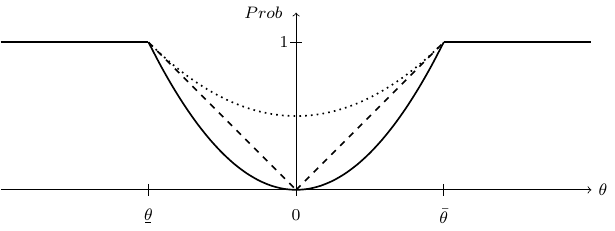}
	\caption{Probability of events in an adversarial equilibrium where players have symmetric features. The solid line depicts the probability that the senders deliver identical and truthful reports. The dashed and the dotted lines depict the probability that the decision maker fully learns the state and that she selects her preferred alternative, respectively.}
	\label{fig:probrevelation}
\end{figure}

Results in this section show the prominence of outcomes with partial revelation. This is in contrast to related work showing that costly talk generates fully revealing equilibria \citep{kartik2007,emons2009accuracy,ottaviani2006}, unless the state space is bounded \citep{kartik2009}. The model's structure plays a key role in determining how much information is transmitted in equilibrium. An important distinctive feature of this model is that the decision maker's action space is discrete, whereas the state space is continuous. In other words, the number of possible contingencies is higher than the number of alternatives. Specifically, the decision maker's choice is binary, reflecting the judicial or quasi-judicial nature of the problem. Intuitively, when the action space is binary, senders want to persuade the decision maker that the state is just sufficiently high (or low). Since different reports come at a different cost, some pooling occurs on path, preventing the full revelation of the state. Differently, when the action space is continuous, higher (lower) reports can induce higher (lower) actions, allowing for separating equilibria.

\section{An Example and Extensions}\label{sec:example}

\subsection{Example: Symmetric Environments}\label{sec:symmetric}

In what follows, I provide an example where senders have similar characteristics and the state is symmetrically distributed. This environment is an important benchmark because it deals with situations where no sender has an ex-ante advantage. In addition, it gives us a closed-form solution for senders' equilibrium strategies and supports. The following definition formalizes what is meant by a symmetric environment.

\begin{definition}\label{def:symm}
	In a symmetric environment,
	\begin{itemize}
		\item[i)] the state is symmetrically distributed around zero, i.e., $f(\theta)=f(-\theta)$ for all $\theta\in\Theta$;
		\item[ii)] $k_j C_j(r,\theta)=k C(r,\theta)$ for $j\in\{1,2\}$, where $k>0$ is finite and $C(\cdot,\cdot)$ satisfies $C(\theta+x,\theta)=C(\theta-x,\theta)$ for every $\theta\in\Theta$ and $x\in\mathbb{R}$; and
		\item[iii)] payoffs satisfy\footnote{By definition of threshold $\tau_j$ (see Section~\ref{sec:model}), this last condition implies that $\tau_2=-\tau_1$.} $u_\dm(\theta)=-u_\dm(-\theta)$ and $u_1(\theta)=-u_2(-\theta)$ for all $\theta\in\Theta$.
	\end{itemize} 
Conditions i) to iii) are in addition to the assumptions in Section~\ref{sec:model}.
\end{definition}

In symmetric environments the two senders differ only because they have conflicting interests. In other words, there is no particular reason why the decision maker should give more importance to the report of one sender than to that of the other. Intuition would suggest that, in a symmetric environment, the decision maker should assign the burden of proof equally between the senders. The next corollary confirms that this intuition is indeed correct in an adversarial equilibrium. 

\begin{restatable}{corollary}{symmetric}\label{cor:symmetry}
	In an adversarial equilibrium of a symmetric environment, $s(r)=-r$ for every $r\in \left[\ubar r_2(0),\bar r_1(0)\right]$.
\end{restatable}

The above corollary shows that, in adversarial equilibria of symmetric environments, the decision maker follows the most extreme recommendation. The burden of proof is equally distributed between the senders through a swing report function that is linear even though some fundamentals, e.g., the cost functions, may be non-linear. Moreover, Corollary~\ref{cor:symmetry} implies that adversarial equilibria naturally have symmetric strategies in symmetric environments.\footnote{In adversarial equilibria, misreporting does not take place outside the set $\left[\ubar r_2(0),\bar r_1(0)\right]$ (see Lemma~\ref{lemma:cutoffs}). Corollary~\ref{cor:symmetry} is reminiscent of results in all-pay auctions with complete information, where it is shown that only symmetric solutions exist with two bidders \citep{baye1996all}.}

With an explicit solution to the swing report function, we obtain a natural closed-form solution to the senders' equilibrium strategies and supports. In applications this is particularly useful because in similar environments, such as in contests, typically little is known about mixed-strategy equilibria except in some special cases (see \citealp{siegel2009,levine2019success}).

I can now use this closed-form solution to examine the determinants and the characteristics of the senders' misreporting behavior. The shape of the cost function, in particular its convexity/concavity, determines whether senders are more likely to deliver small lies or large misrepresentation or the other way around. Recall that $S_j(\theta)$ is the support of sender~$j$'s strategy in state $\theta$. By Step~\ref{step:strategies} (Appendix~\ref{app:aestrat}) and Corollary~\ref{cor:symmetry} we obtain that, in a symmetric environment, misreporting behavior is described by the following partial density, for $j\in\{1,2\}$ and $j\neq i$,
\begin{equation}\label{eq:partialdensity}
\psi_j(r_j,\theta)=\frac{k}{-u_i(\theta)}\frac{d C(-r_j,\theta)}{d r_j},
\end{equation}
where $\psi_j:\Theta^2\to\mathbb{R}^+_0$ has support in $S_j(\theta)\setminus\{\theta\}$ (see Appendix~\ref{app:notation} for more details). From \eqref{eq:partialdensity} we can see that, if $C(\cdot,\cdot)$ is strictly convex, then we have $d\psi_1(r_1,\theta)/dr_1>0$ for all $\theta\in S_1(\theta)\setminus\{\theta\}$ and $d\psi_2(r_2,\theta)/dr_2<0$ for all $\theta\in S_2(\theta)\setminus\{\theta\}$. This means that, conditional on misreporting, senders are more likely to deliver large misrepresentations of the state than small lies. By contrast, when senders have concave costs, misreports that are closer to the truth are more likely to be delivered than large misrepresentations. This observation may seem counter-intuitive at first. However, a sender's misreporting density $\psi_j$ is directly affected by his opponent's marginal costs (see Step~\ref{step:strategies} in Appendix~\ref{app:aestrat}). Equilibrium conditions require senders to be indifferent about misreporting slightly more. The marginal cost of misreporting must be compensated by a proportionally higher probability of inducing the sender's preferred alternative. Since a sender's density $\psi_j$ is proportional to the marginal costs of his competitor, $dC_{-j}(s(r_j),\cdot)/dr_j$, the former is higher when the latter is higher. In this symmetric environment, $\psi_j(r_j,\cdot)\propto dC(-r_j,\cdot)/dr_j$. Figure~\ref{fig:reportconvexity} depicts sender~1's misreporting behavior in a symmetric environment for different concavities of the misreporting cost function.

\begin{figure}
	\centering
	\includegraphics[width=0.5\linewidth]{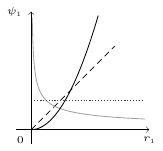}
	\caption{The partial probability density $\psi_1$ as a function of $r_1$ in state $\theta=0$ and for different costs structures. The environment is symmetric, and the cost functions are $C_j(r_j,\theta)=|r_j-\theta|^{exp}$. With quadratic loss costs, $exp=2$, the density $\psi_1$ grows linearly as reports get further away from the truth (black dashed line). With absolute value linear costs, $exp=1$, every misreport in the support has the same partial density (black dotted line). With concave costs, $exp=1/2$, small misrepresentations are more likely than large lies (grey solid line), and when $exp=3$ the opposite is true (black solid line).}
	\label{fig:reportconvexity}
\end{figure}

\subsection{Discussion and Extensions}\label{sec:discussion}

In this section, I discuss a number of extensions and variations of the baseline model to examine the robustness of the results. 

\noindent{\bf A single sender.} Consider a variant of the model presented in Section~\ref{sec:model} where only sender~1 communicates with the decision maker. The rest of the model remains as before. In this monopolistic setting, there is a continuum of PBE. Notably, there is a receiver-efficient equilibrium where sender~1 plays the same strategy as in the fully revealing equilibrium depicted in Figure~\ref{fig:fre}, Section~\ref{sec:fre}. That is, $\rho_1(\theta)=\bar r_1(0)$ for all $\theta\in[0,\bar r_1(0)]$, and $\rho_1(\theta)=\theta$ otherwise. In addition, there are equilibria that are not receiver-efficient. In the sender-preferred one, the monopolistic sender achieves persuasion by pooling states around the decision maker's threshold $\tau_{dm}$. Receiver-efficient outcomes do not withstand refinements even when there is only a single monopolistic sender.\footnote{I study more in details this monopolistic setting in \cite{vaccari2021influential}, where, among other results, I show that the sender-preferred equilibrium is the only equilibrium of the monopolistic game to be perfectly sequential \citep{grossman1986perfect}. In addition, every non-revealing equilibrium defeats the receiver-efficient equilibrium \citep{mailath1993belief}. Among the equilibria passing the Intuitive Criterion test, only the sender-preferred one is undefeated.} Figure~\ref{fig:mono} shows the sender's reporting strategy in the receiver-efficient and in the sender-preferred equilibrium of the monopolistic setting.

In equilibria of the monopolistic game, the sender truthfully reveals extreme states and \emph{pools} moderate states that are close to the decision maker's threshold $\tau_{dm}$. This communication pattern is similar in adversarial equilibria of the competitive game, where truthful revelation occurs in extreme states while misreporting and persuasion occur in moderate states.\footnote{Adversarial equilibria also feature a probabilistic revelation of the state, including the moderate ones. This phenomenon does not take place in equilibria of the monopolistic game, where information of moderate states is always \emph{jammed}.} Equilibria of the monopolistic game are all in pure strategies. Intuitively, a monopolistic sender delivers the cheapest report that induces the selection of his favourite alternative, whenever this is feasible. This logic does not readily extend to competitive settings: pure-strategy equilibria of the baseline model with competition exist, but they do not survive refinements.

\begin{figure}
	\centering
	\begin{subfigure}{.5\textwidth}
		\centering
		\includegraphics[width=1\linewidth]{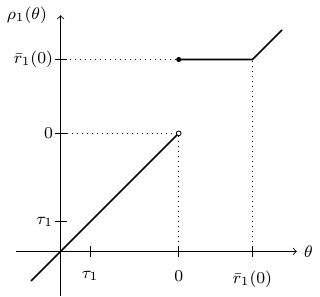}
		\label{fig:monop_fre}
	\end{subfigure}%
	\begin{subfigure}{.5\textwidth}
		\centering
		\includegraphics[width=1\linewidth]{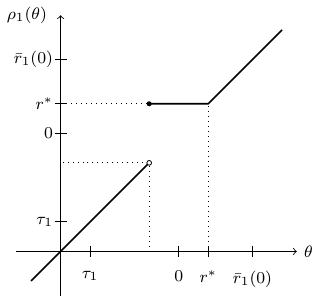}
		\label{fig:monop_sp}
	\end{subfigure}
	\caption{The panels illustrate the sender's reporting strategy in two different equilibria of the monopolistic game. In the left panel the equilibrium is receiver-efficient, and the strategy of sender~1 is identical to the one he plays in the fully revealing equilibrium of the game with two senders, discussed in Section~\ref{sec:fre}. The right panel depicts the sender-preferred equilibrium, which is not receiver-efficient. There, the sender pools states around the decision maker's threshold $\tau_{dm}=0$. Upon observing the pooling report $r^*$, the decision maker is indifferent between the two alternatives.}
	\label{fig:mono}
\end{figure}

{\bf Withholding information.} The baseline model does not allow senders to withhold their private information or, analogously, to stay silent. Here, I discuss an extension that accommodates for information withholding. Let the senders' report space be $\Theta\cup \{\varnothing\}$, where $\varnothing$ corresponds to not delivering any report, withholding information, remaining silent, etc. Assume that, for every state $\theta\in\Theta$, sender $j$'s cost of withholding information\footnote{Withholding information may be a costly activity. For example, a sender that is known to be informed about the realized state may suffer from a reputation loss if he refuses to communicate. Alternatively, withholding information may require an active act of suppression, which takes up resources.} is $C_j(\varnothing,\theta)=c_j\geq 0$. The rest of the model remains as described in Section~\ref{sec:model}. Hereafter, I use ``silence'' and ``withholding information'' interchangeably.\footnote{\cite{emons2019strategic} study an adversarial persuasion game where silence is costless while reporting is always costly. In their setting, some states are never revealed in equilibrium. This is in contrast with adversarial equilibria of the current setup, where almost every state is revealed with positive probability.}

Adversarial equilibria can be extended to this augmented environment in a natural way. Consider the senders' reporting strategies in an adversarial equilibrium of the model where silence is not allowed; introduce the possibility of withholding information, and suppose that the decision maker is skeptical about silence: when sender~1 (2) withholds information and sender~2 (1) claims that the state is negative (positive), the decision maker selects the negative (positive) alternative.\footnote{Formally, beliefs $p$ are such that $\beta(\varnothing,r_2)=\n$ for every $r_2\leq 0$, and $\beta(r_1,\varnothing)=\p$ for every $r_1\geq 0$.} Since reporting truthfully is at least as cheap but also at least as good as withholding information, the original \emph{adversarial} reporting strategies constitute an equilibrium also in this model's variant. The same holds true for the receiver-efficient equilibrium depicted in Section~\ref{sec:fre}.

The possibility of withholding information at no cost challenges the finding that receiver-efficient outcomes lack robustness. Intuitively, introducing costless silence restores the possibility of babbling, allowing the decision maker to neglect reports delivered by a silent sender. In this case, there is a receiver-efficient equilibrium where the decision maker obtains information only from one sender, as the other one is always ignored and thus remains silent. After observing an unexpected report, the decision maker cannot extract the information she needs from the silent sender. This equilibrium breaks down once information withholding comes at a positive cost, no matter how small: conditional on the outcome, senders would rather report truthfully than stay silent. As a result, costly silence preserves the results in Section~\ref{sec:fre}. The next proposition summarizes the takeaways from this model extension.

\begin{restatable}{proposition}{silence}\label{prop:silence}
	In the game where senders can withhold information: (i) there are no unprejudiced REE when silence is costly for both senders; (ii) there is an unprejudiced REE when silence is costless for at least one sender.
\end{restatable}

\noindent {\bf More than two senders.} Here, I discuss a version of the baseline model with more than two biased senders. Denote with $N$ the set of senders and with $r$ the collection of all reports, $r=\{r_j\}_{j\in N}$. Suppose that at least one sender has a negative threshold and at least one sender has a positive threshold $\tau_j$. The rest of the model is as before, including that the senders have common and perfect information. Like in the baseline model, there is competition between opposed-biased senders. The main difference is that now there is also a potential problem of coordination between like-biased senders.

Communication games with three or more perfectly informed senders admit fully revealing equilibria.\footnote{For cheap talk games, see~\cite{battaglini2004}. The same logic used there applies to this costly talk setting. In truthful FRE, senders report truthfully every state, that is, $\rho_j(\theta)=\theta$ for every $\theta\in\Theta$ and $j\in N$. Unilateral deviations are always detected, and thus cannot impede full revelation. These equilibria have unprejudiced beliefs: when all senders agree except for one, the decision maker is sure that only the disagreeing sender has deviated from the prescribed truthful strategy.} To check if these equilibria are sensible in a similar way as adversarial ones, we first need to extend conditions \eqref{eq:c} and \eqref{eq:d} to settings with more than two senders. Since there is no obvious way to do so, I focus on equilibria that satisfy this alternative condition: whenever possible, the decision maker believes that the state is somewhere in between the highest report delivered by senders with a positive threshold and the lowest report delivered by senders with a negative threshold.\footnote{Denote by $r_L$ the highest report among those delivered by senders with $\tau_j>0$, and by $r_H$ the lowest report among those delivered by senders with $\tau_j<0$. This alternative condition requires the decision maker to believe that the state is surely between $r_L$ and $r_H$ whenever $r_L<r_H$.} This alternative condition is relatively simple and mimics some implications of \eqref{eq:c} and \eqref{eq:d}.

The fully revealing equilibrium in truthful strategies is the only receiver-efficient equilibrium satisfying the alternative condition.\footnote{Since the proof of Lemma~\ref{lemma:pureREE} extends to this variant of the model with more than two senders, it follows that the truthful equilibrium is also the only robust equilibrium in pure strategies. The alternative condition does not rule out the plausibility of fully revealing and pure-strategy equilibria.} However, senders' coordination---which is possible when there are more than two competing senders---hinders fully revealing outcomes. This is the case, for example, when senders can collude or engage in non-binding pre-play communication. Intuitively, like-biased senders may agree to deviate from truthful reporting in a mutually beneficial and self-enforcing way. \cite{bernheim1987coalition} propose the notion of ``coalition-proof equilibrium'' as a solution concept for situations where players can freely discuss their strategies but cannot make binding commitments. The next result confirms the intuition that truthful equilibria are not immune to group deviations.\footnote{The type of group deviations considered by the notion of coalition-proofness is consistent with the model because it preserves its non-cooperative nature. Proposition~\ref{prop:coalition} does not require the application of any additional condition, and it applies to cheap talk environments as well.} Non-revealing equilibria remain important in a large class of economic environments beyond the two-senders case.

\begin{restatable}{proposition}{coalition}\label{prop:coalition}
Fully revealing equilibria in truthful strategies of the game with $n>2$ senders are not coalition-proof.
\end{restatable}

{\bf Uncertain preferences.} All aspects of the baseline model are common knowledge except for the realized state, which is known only to the senders. Additional uncertainty may discourage misreporting, result in more transmission of information, and potentially restore the plausibility of receiver-efficient equilibria. This may be the case, for example, when senders are uncertain about the preferences of the decision maker. I show that such type of uncertainty does not necessarily result in robust receiver-efficient equilibria, even when the state and report spaces are unbounded. 

Consider a variant of the model where senders are uncertain about the decision maker's threshold $\tau_{dm}$, which is distributed according to the common knowledge pdf $f_{dm}$. The distribution $f_{dm}$ has full support in $\left[\ubar t,\bar t\right]$, where $\tau_2>\bar t>\ubar t>\tau_1$. The support of $f_{dm}$ ensures that senders are always in competition against each other and have opposed biases. By contrast, the decision maker knows her own preferences. The rest of the model remains as before. Suppose that $\Theta\supset \left[\ubar t,\bar t\right]$ and senders' preferences are such that $\ubar r_2(\bar t)< \ubar t<\bar t \leq \bar r_1(\ubar t)$. This last assumption ensures that senders can deliver unambiguous recommendations, as in the baseline model.\footnote{Specifically, it allows sender~1 (2) to deliver a report $r_1\geq \bar t$ ($r_2\leq \ubar t$) even when the state is such that the decision maker unambiguously prefers the negative (positive) alternative. Similarly, the state space in the baseline model allows sender~1 to deliver reports $r_1\geq \bar r_1(0)$ and sender~2 to deliver $r_2\leq \ubar r_2(0)$.} Denote by $\Gamma'$ the game resulting from this variant of the baseline model. The next lemma shows that this model extension does not restore the existence of robust receiver-efficient equilibria.
\begin{restatable}{lemma}{uncertain}\label{lemma:uncertain}
	If a receiver-efficient equilibrium of $\Gamma'$ exists, then it is not unprejudiced.
\end{restatable}

\section{Concluding Remarks}\label{sec:conclusion}

This paper presents a model of adversarial communication between two perfectly informed senders and one uninformed decision maker. Senders can misreport information at a cost that is tied to the magnitude of misrepresentation. Misreporting costs represent, e.g., direct costs due to technological constraints or indirect costs due to expected reputation damages. The setting considered here covers several applications, including electoral campaigns, contested takeovers, lobbying, informative advertising, and judicial decision making. The main results show that equilibria where the decision maker obtains her complete-information payoff do not withstand refinements. The imposition of a minimal and natural beliefs structure generates robust equilibria with desirable properties where information transmission is only partial.

The analysis conducted in this paper provides a tractable and appealing approach to studying adversarial communication of information that is neither fully verifiable nor entirely ``cheap.'' Previous work shows that communication games typically display fully revealing equilibria when there are multiple senders or when misreporting is costly. By contrast, the analysis conducted here shows that revealing outcomes are not robust when the decision maker's problem is judicial, even with two competing senders. This finding enables us to model situations where information asymmetries can effectively result in persuasion, even when misreporting is costly and senders' reports can be cross-validated. 

Results obtained in this paper can be readily applied to analyze the informative value of judicial procedures. In a seminal paper, \cite{shin1998adversarial} conjectures that the assumption of full verifiability plays a key role in determining the superiority of adversarial over inquisitorial judicial procedures.\footnote{For example, ``[\ldots] violations of the verifiability assumption will be an important limiting factor in qualifying our findings in favor of the adversarial procedure'' \citep[p.~403]{shin1998adversarial}. Under the adversarial procedure, two parties with conflicting interests make their case to an uninformed decision maker. By contrast, the inquisitorial procedure requires the decision maker to adjudicate based only on her acquired information. The comparison is between a model where the decision maker gets information from two competing senders and one where she acquires it through an investigation.} We can validate this conjecture by modeling adversarial procedures with the framework studied in this paper. Findings in Sections~\ref{sec:fre} and \ref{sec:mixed} tell us that we can reject receiver-efficient outcomes by imposing a natural language structure or on robustness grounds. Once we are left with non-revealing outcomes, there is space for inquisitorial procedures (i.e., information acquisition) to dominate adversarial procedures. The conjecture of \cite{shin1998adversarial} is thus correct for any finite intensity of misreporting costs.\footnote{In addition to the verifiability assumption, there are other modeling differences between my setting and that of \cite{shin1998adversarial}: first, I assume that the senders are always perfectly informed about the state, while in \cite{shin1998adversarial} they may be uninformed or observe a noisy signal of the state; second, I consider a decision maker that is less informed than the senders, while in \cite{shin1998adversarial} every player is, on average, equally informed. In my setting, these differences give a relative advantage to the adversarial procedure, and therefore add further force to the potential superiority of the inquisitorial procedure.}


\pagebreak

\appendix

\section{Appendix}\label{sec:app}

In the main text, I assume for simplicity that the state and report space coincide with the real line, i.e., $\Theta=\mathbb{R}$. However, the analysis carried in this paper only requires the state space to be large enough. Hereafter, I assume that
\begin{displaymath}
	\Theta \supseteq \left[\ubar r_2(0),\bar r_1(0)\right].
\end{displaymath}
This assumption ensures that the information senders can transmit is not artificially bounded by restrictions in the reports that they can deliver.\footnote{For example, if $\max\Theta<\bar r_1(0)$, then sender~1 cannot deliver reports that are strictly dominated by truthful reporting when the realized state is strictly negative.}

\monotlemma*
\begin{proof}
	Consider a PBE satisfying \eqref{eq:m} and consider a state $\theta \geq \tau_1$. For sender~1, every report $r_1<\theta$ is dominated by truthful reporting because $C_1(r_1,\theta)>0=C_1(\theta,\theta)$ and (by~\eqref{eq:m}) $U_\dm(\theta,r_2)\geq U_\dm(r_1,r_2)$ for every $r_2\in \Theta$. Therefore, it must be that $r_1\notin S_1(\theta)$ for all $r_1<\theta$ and $\theta\geq \tau_1$. A similar argument applies to sender~2 and to states $\theta\leq\tau_j$, $j\in \{1,2\}$.
\end{proof}

\subsection{Receiver-efficient and Pure-strategy Equilibria}\label{app:REE}

\lemmababbling*
\begin{proof}
	Suppose that there is an equilibrium in which the decision maker's strategy is independent of sender $j$'s reports. Since misreporting is costly, sender $j$'s best reply is to report truthfully in every state. As a result, it is not sequentially rational for the decision maker to ignore sender $j$'s reports, contradicting that this is an equilibrium. The same line of logic applies to equilibria in which the decision maker's strategy is independent of both senders' reports. Therefore, there are no babbling equilibria.
	
	Suppose by way of contradiction that there exists an equilibrium where misreporting never occurs, that is, where $\rho_1(\theta)=\rho_2(\theta)=\theta$ for every $\theta\in\Theta$. Consider such a truthful equilibrium and a state $\theta=\epsilon>0$, where $\epsilon$ is small enough. To discourage deviations, off path beliefs must be such that $\beta(\epsilon,-\epsilon)=\p$. However, there always exists an $\epsilon>0$ such that, when the state is $\theta=-\epsilon$, sender~1 can profitably deviate from the prescribed truthful strategy by reporting $r_1=\epsilon$, as $u_1(-\epsilon)>k_1 C_1(\epsilon,-\epsilon)$. This contradicts our supposition that there exists an equilibrium where misreporting never occurs.
\end{proof}

\pureREE*
\begin{proof}
	Consider a pure-strategy equilibrium and suppose that it is not receiver-efficient, e.g., because $\beta(\rho_1(\theta'),\rho_2(\theta'))=\n$ for some $\theta'\geq 0$. In equilibrium, senders never engage in misreporting to implement their less preferred alternative with certainty, and therefore it must be that $\rho_1(\theta')=\theta'$. Posterior beliefs $p$ must be such that $\beta(r_1,\rho_2(\theta'))=\n$ for all $r_1\in(\ubar r_1(\theta'),\bar r_1(\theta'))$, as otherwise sender~1 would have a profitable deviation. The pair of reports $(\theta',\rho_2(\theta'))$ can induce $\n$ only if $(\rho_1(\theta''),\rho_2(\theta''))=(\theta',\rho_2(\theta'))$ for some $\theta''<0$. There is no $\theta\in[\tau_1,0)$ such that sender $1$ would misreport by delivering $r_1=\theta'\geq 0$ to implement $\n$, and therefore it must be that $\theta''<\tau_1$. Since there is always $r_1'\in(\ubar r_1(\theta'),\theta')$ such that $C_1(r_1',\theta'')<C_1(\theta',\theta'')$ and $\beta(r_1',\rho_2(\theta''))=\n$, sender~1 has a profitable deviation in state $\theta''$, contradicting that there exists a pure-strategy equilibrium that is not receiver-efficient.	
	
	Now consider a REE and suppose that it is not in pure strategies, but that there is a state $\theta'\in \Theta$ and a sender $j\in\{1,2\}$ such that $S_j(\theta')\supseteq\{r_j',r_j''\}$, with $r_j'\neq r_j''$. Since in a REE we have that $\beta(r_1',r_2')=\beta(r_1'',r_2'')$ for every $r_i',r_i''\in S_i(\theta)$, $i\in\{1,2\}$, it must be that $C_j(r_j',\theta')=C_j(r_j'',\theta')$. By Lemma~\ref{lemma:monot}, this is possible only if $r_j'=r_j''$, contradicting that there exists a REE that is not in pure strategies.
\end{proof}

\nonREE*

\begin{proof}
	In a REE, senders play pure strategies (Lemma~\ref{lemma:pureREE}) and the decision maker always selects her preferred alternative \emph{as if} she has complete information, that is,  $\beta(\rho_1(\theta),\rho_2(\theta))=\p$ for all $\theta\geq 0$ and $\beta(\rho_1(\theta),\rho_2(\theta))=\n$ otherwise. Since misreporting is costly, senders report truthfully in states where their least preferred alternative is implemented: $\rho_2(\theta)=\theta$ for all $\theta\in[0,\tau_2]$ and $\rho_1(\theta)=\theta$ for all $\theta\in[\tau_1,0)$. However, there are no REE where $\rho_j(\theta)=\theta$ for all $\theta\in[\tau_1,\tau_2]$ and $j\in\{1,2\}$; for otherwise, there would always be a state $\theta\in(\tau_1,\tau_2)$ and an off path pair of reports $(r_1,r_2)$, $r_1\neq r_2$, such that a sender can profitably deviate from truthful reporting (see also Lemma~\ref{lemma:babbling}). Therefore, in every REE either sender~1 misreports in some state $\theta\in [0,\tau_2)$, or sender~2 misreports in some $\theta\in(\tau_1,0]$, or both. 
	
	Consider now a REE where $\rho_1(\theta')\neq\theta'$ for some $\theta'\in[0,\tau_2)$. By Lemma~\ref{lemma:monot}, we have that $\rho_1(\theta')>\theta'$. To support the equilibrium, off path beliefs $p$ must be such that $\beta(r_1,\theta')=\n$ for all $r_1\in[\theta',\rho_1(\theta'))$ and $\beta(\rho_1(\theta''),r_2)=\p$ for all $r_2\in(\ubar r_2(\theta''),\theta'']$ and $\theta''\in[\theta',\tau_2)$. This implies that there must be an open set $S$ of non-negative states such that $\rho_1(\theta''')\geq\rho_1(\theta')>\theta'''=\rho_2(\theta''')$ for all $\theta'''\in S$. It follows that, for every $\theta'''\in S$, the pair of reports $(\theta''',\theta''')$ is off path. By Lemma~\ref{lemma:monot}, and since $\rho_2(\theta)=\theta$ for all $\theta\in[0,\tau_2]$ and $\rho_1(\theta)=\theta$ for all $\theta\in[\tau_1,0)$, we have that posterior beliefs $p$ are unprejudiced (Definition~\ref{def:unprejudiced}) only if $p(\theta\,|\,\theta''',\theta''')=0$ for all $\theta<0$. Therefore, unprejudiced beliefs imply that $\beta(\theta''',\theta''')=\p$, and therefore sender~1 can profitably deviate by reporting the truth in state $\theta'''\in S$. A similar argument applies to REE where $\rho_2(\theta')\neq\theta'$ for some $\theta'\in(\tau_1,0]$. Therefore, there are no REE (and, by Lemma \ref{lemma:pureREE}, no pure-strategy equilibria) with unprejudiced beliefs.
\end{proof}

\epsilonrobust*

\begin{proof}
	Consider the posterior beliefs $p_{G,\varepsilon}$ that the strategies $\phi_j$ of a PBE (see Section~\ref{sec:mixed} for the notation used to describe mixed strategies) induce in an $\varepsilon$-perturbed game for some distribution $G$ and sequence $\varepsilon^n$, i.e.,
	\begin{multline*}
		p_{G,\varepsilon}(\theta\,|\,r_1,r_2) =f(\theta)\frac{p_{G,\varepsilon}(r_1,r_2\,|\,\theta)}{p_{G,\varepsilon}(r_1,r_2)} \\
		=\frac{f(\theta)\left[\varepsilon_1 \varepsilon_2 g_1(r_1) g_2(r_2) + \varepsilon_1(1-\varepsilon_2)g_1(r_1)\phi_2(r_2,\theta) + (1-\varepsilon_1)\varepsilon_2 g_2(r_2) \phi_1(r_1,\theta)\right]}{\varepsilon_1 \varepsilon_2 g_1(r_1) g_2(r_2) +  \varepsilon_1(1-\varepsilon_2)g_1(r_1) \int_\Theta f(\theta)\phi_2(r_2,\theta)d\theta + (1-\varepsilon_1)\varepsilon_2 g_2(r_2) \int_\Theta f(\theta)\phi_1(r_1,\theta)d\theta}.
	\end{multline*}
	As $\varepsilon^n \to 0^+$ the event in which both reports are wrongly delivered or observed becomes negligible, and thus we have that $p_{G,\varepsilon} \to p_{G,0^+}$, where
	\begin{equation}\label{eq:epsilonbayes}
		p_{G,0^+}(\theta\,|\,r_1,r_2)  =\frac{f(\theta)\left[\varepsilon_1 g_1(r_1) \phi_2(r_2,\theta)+\varepsilon_2 g_2(r_2)\phi_1(r_1,\theta)\right]}{\varepsilon_1 g_1(r_1)\int_\Theta f(\theta)\phi_2(r_2,\theta)d\theta + \varepsilon_2 g_2(r_2)\int_\Theta f(\theta)\phi_1(r_1,\theta)d\theta}. 
	\end{equation}
	By \eqref{eq:epsilonbayes} we obtain that, for any distribution $G$ with full support and any sequence $\varepsilon^n\to 0^+$, $p_{G,0^+}(\theta\,|\,r_1,r_2)>0$ if and only if $\phi_j(r_j,\theta)>0$ for some $j\in\{1,2\}$. By Definition~\ref{def:unprejudiced2} (and hence even by Definition~\ref{def:unprejudiced}) we get that the limit beliefs $p_{G,0^+}$ are unprejudiced, and therefore every PBE of the game described in Section~\ref{sec:model} (and hence every equilibrium) that is $\varepsilon$-robust is supported by unprejudiced beliefs.\footnote{Notice that the proof of Lemma~\ref{lemma:epsilonunprejudiced} readily extends to an $n$-sender version of the game, for any finite $n\geq2$. In particular, given a profile of reports $(r_1,\ldots,r_n)$ and a set of senders $N=\{1,\ldots,n\}$, we have that $p_{G,0^+}(\theta\,|\,r_1,\ldots,r_n)>0$ if and only if $\phi_j(r_j,\theta)>0$ for $n-1$ senders. This is consistent with the idea behind unprejudiced beliefs, where the decision maker conjectures that deviations are individual.}
\end{proof}

\reenotde*

\begin{proof}
	Consider a REE and suppose that it is also an AE. Conditions \eqref{eq:c} and \eqref{eq:d} apply, implying that $\beta(r_1,r_2)=\p$ for every pair of reports such that $r_j\geq 0$ for $j\in\{1,2\}$. Consider a state $\theta'\geq 0$. Since the equilibrium is receiver-efficient, we have that $\beta(\rho_1(\theta'),\rho_2(\theta'))=\p$. To prevent deviations by sender~2, beliefs must yield $\beta(\rho_1(\theta'),r_2)=\p$ for every $r_2\in(\ubar r_2(\theta'),\bar r_2(\theta'))$. Because of \eqref{eq:d}, it must be that $\rho_2(\theta')=\theta'$. It follows that $\rho_1(\theta')=\theta'$, as we have noticed before that $\beta(\theta',\theta')=\p$ for every $\theta'\geq 0$. A similar logic applies to states $\theta'<0$, implying that both senders always report truthfully. This is in contradiction with Lemma~\ref{lemma:babbling}. Therefore, a REE cannot be an AE.
\end{proof}

\subsection{Adversarial Equilibria}\label{app:mixed}

\subsubsection{Notation for Mixed Strategies}\label{app:notation}

Before analyzing adversarial equilibria, I first introduce some useful notation. To describe mixed strategies, I use a ``mixed'' probability distribution $\phi_j(r_j,\theta)$ that, for every state $\theta$, assigns a mixed probability density to report $r_j$ by sender $j$. This specification allows me to describe the senders' reporting strategies as mixed random variables whose distribution can be partly continuous and partly discrete.\footnote{Mixed-type distributions that have both a continuous and a discrete component to their probability distributions are widely used to model zero-inflated data such as queuing times. For example, the ``rectified Gaussian'' is a mixed discrete-continuous distribution.} I denote the probability that sender $j$ reports truthfully in state $\theta$ with $\alpha_j(\theta)$. I denote the probability (density) that sender $j$ misreports by delivering report $r_j$ in state $\theta$ with $\psi_j(r_j,\theta)$. Together, the partial probability functions $\alpha_j(\theta)$ and $\psi_j(\cdot,\theta)$ determine the mixed probability distribution $\phi_j(\cdot,\theta)$ in a way that I explain more in details as follows. The CDFs of $\phi_j$ and $\psi_j$ are $\Phi_j(r_j,\theta)$ and $\Psi_j(r_j,\theta)$, respectively.

Formally, I partition the support $S_j(\theta)$ of each sender's strategy in two subsets, $\mathcal{C}_j(\theta)$ and $\mathcal{D}_j(\theta)$. To represent atoms in $\phi_j(\theta)$, I define a partial probability density function $\alpha_j(\cdot,\theta)$ on $\mathcal{D}_j(\theta)$ such that $0\leq \alpha_j(r_j,\theta)\leq 1$ for all $r_j\in \mathcal{D}_j(\theta)$, and $\hat \alpha_j(\theta)=\sum_{r_j\in \mathcal{D}_j(\theta)} \alpha_j(r_j,\theta)$. By contrast, the continuous part of the distribution $\phi_j(\theta)$ is described by a partial probability density function $\psi_j(\cdot,\theta)$ on $\mathcal{C}_j(\theta)$ such that $\int_{r_j\in \mathcal{C}_j(\theta)}\psi_j(r_j,\theta)d\theta=1-\hat\alpha_j(\theta)$. I set $\alpha_j(r',\theta)=0$ for all $r'\notin \mathcal{D}_j(\theta)$ and $\psi_j(r'',\theta)=0$ for all $r''\notin \mathcal{C}_j(\theta)$. 

As we shall see (Lemma~\ref{lemma:noatom2} and Step~\ref{step:atoms}), in every adversarial equilibrium $\mathcal{D}_j(\theta)=\{\theta\}$ for all $\theta\in\Theta$ and $j\in\{1,2\}$. Therefore, for ease of notation, I set $\alpha_j(\theta)\equiv\alpha_j(\theta,\theta)=\hat\alpha_j(\theta)$. The score $\alpha_j(\theta)$ thus represents the probability that sender $j$ reports truthfully in state $\theta\in\Theta$. The partial probability density functions\footnote{Under this specification, even the ``mass'' $\alpha_j(\cdot)$ is a partial probability ``density.''} $\alpha_j(\theta)$ and $\psi_j(\cdot,\theta)$ determine the ``generalized'' density function $\phi_j(\theta)$ through the well-defined mixed distribution
\[
\phi_j(x,\theta)= \delta(x-\theta)\alpha_j(\theta) + \psi_j(x,\theta),
\] 
where $\delta(\cdot)$ is the Dirac delta ``generalized'' function.\footnote{The Dirac delta $\delta(x)$ is a generalized function such that $\delta(x)=0$ for all $x\neq 0$, $\delta(0)=\infty$, and $\int_{-\epsilon}^{\epsilon} \delta(x)dx=1$ for all $\epsilon>0$.}

A mixed-strategy for sender $j$ is a mixed probability measure $\phi_j(\theta):\Theta\to\Delta(\Theta)$ with support $S_j(\theta)$. I indicate with $\phi_j(r_j,\theta)$ the mixed probability assigned by $\phi_j(\theta)$ to a report $r_j$ in state $\theta$ that satisfies
\[
\int_{r_j\in S_j(\theta)}\phi_j(r_j,\theta)d r_j=\alpha_j(\theta)+\int_{r_j\in \mathcal{C}_j(\theta)}\psi_j(r_j,\theta)dr_j=1.
\]
Sender $j$'s expected utility from delivering $r_j$ when the state is $\theta$ in an adversarial equilibrium $\omega$ is $W^\omega_j(r_j,\theta)$.

\subsubsection{Supporting Lemmata}

\begin{Alemma}\label{lemma:swing}
	In an adversarial equilibrium, every report $r\in \left[\ubar r_2(0),\bar r_1(0)\right]$ has a swing report $s(r)\in \left[\ubar r_2(0),\bar r_1(0)\right]$ such that: (i) if $r \gtrless  0$ then $s(r)\lessgtr0$ and $s(0)=0$; (ii)  $s(s(r))=r$; (iii) for every $r\in \left[\ubar r_2(0),\bar r_1(0)\right]$, $d s(r)/d r<0$; (iv) $s\left( \bar r_1(0) \right)=\ubar r_2(0)$.
\end{Alemma} 

\begin{proof}
	Consider a report $r_1$ by sender~1 such that $r_1\in(0,\bar r_1(0)]$. By \eqref{eq:c} and \eqref{eq:d} we obtain that $U_\dm(r_1,\ubar r_2(0))<0< U_\dm(r_1,0)$, and therefore there exists $r_2\in[\ubar r_2(0),0)$ such that $U_\dm(r_1,r_2)=0$. Thus, $r_2=s(r_1)$. A similar argument holds for a report $r_2\in[\ubar r_2(0),0)$. It follows that, for every $r\in \left[\ubar r_2(0),\bar r_1(0)\right]$, there exists $s(r)\in \left[\ubar r_2(0),\bar r_1(0)\right]$ such that if $r>0$ then $s(r)<0$, and if $r<0$ then $s(r)>0$. From \eqref{eq:c} and Definition~\ref{def:swing} we obtain that $s(0)=0$ and $s\left( \bar r_1(0) \right)=\ubar r_2(0)$. From Definition~\ref{def:swing} and point (i) we get that if $r'=s(r)$ then $r=s(r')$, and therefore $s(s(r))=r$. By applying the implicit function theorem and \eqref{eq:d} to $s(r)$, we obtain that for every $r\in \left[\ubar r_2(0),\bar r_1(0)\right]$, $d s(r)/d r<0$.
\end{proof}


\begin{Alemma}\label{lemma:cutoffs}
	In an adversarial equilibrium, truthful cutoffs are such that $\theta_l < 0 < \theta_h$ and $(\theta_l,\theta_h)\subset [\tau_1,\tau_2]\cap \left[\ubar r_2(0),\bar r_1(0)\right]$.	
\end{Alemma}

\begin{proof}
	By Lemma~\ref{lemma:swing} we have that $s\left(\bar r_1(0)\right)=\ubar r_2(0)<0$ and, for every $r\in\left[\ubar r_2(0),\bar r_1(0)\right]$, $d s(r) / dr<0$. Moreover, $d \ubar r_2(\theta) / d\theta>0$ and thus $\ubar r_2(\theta)>\ubar r_2(0)$ for every $\theta>0$. Since $s(0)=0$, there is a state $\theta'\in(0,\bar r_1(0))$ such that $s(\theta')=\ubar r_2(\theta')$. From Definition~\ref{def:swing}, we obtain that $\theta'=\theta_h\in(0,\bar r_1(0))$. Similarly, we get that $\theta_l\in (\ubar r_2(0),0)$. Since $\bar r_1(\tau_1)=\tau_1<0$ and $\ubar r_2(\tau_2)=\tau_2>0$, it follows from Definition~\ref{def:cutoffs} that $(\theta_l,\theta_h)\subset[\tau_1,\tau_2]$.
\end{proof}

\begin{Alemma}\label{lemma:supportscutoffs}
	In an adversarial equilibrium: (i) $S_j(\theta)=\{\theta\}$ for $j\in\{1,2\}$ and all $\theta\notin(\theta_l,\theta_h)$; (ii) $S_j(\theta)$ contains more than one report for $j\in\{1,2\}$ and all $\theta\in(\theta_l,\theta_h)$.
\end{Alemma}

\begin{proof}
	I begin by proving that $S_j(\theta)=\{\theta\}$ for all $\theta\notin(\theta_l,\theta_h)$. 	Consider an AE and a state $\theta\geq\theta_h$. Since by Lemma~\ref{lemma:monot} we have that $\min S_1(\theta)\geq\theta\geq\theta_h$, it must be that $S_2(\theta)=\{\theta\}$ as $s(r_1)\leq \ubar r_2(\theta)$ for every $r_1\in S_1(\theta)$. Since  $\beta(\theta,\theta)=\p$, sender~1 best replies to $r_2=\theta$ with $r_1=\theta$ and therefore $S_1(\theta)=\{\theta\}$ as well. A similar argument applies to states $\theta\leq\theta_l$, completing the first part of the proof. Note that when $\theta=\theta_l$, sender~1 is actually indifferent between reporting $\theta_l$ and $\bar r_1(\theta_l)$. Since this is a measure zero event, which is irrelevant to the analysis that follows, I will consider only the case where $S_1(\theta_l)=\{\theta_l\}$, without any loss of generality. 
	
	I turn now to prove that $S_j(\theta)$ contains more than one report for every $\theta\in(\theta_l,\theta_h)$. Suppose by way of contradiction that $S_1(\theta)=\{r_1\}$ for some $\theta\in(\theta_l,\theta_h)$. By Lemma~\ref{lemma:monot}, we have that $r_1\geq\theta$. Consider first the case where $\theta\leq r_1<0$. In an AE, sender~2 best replies to $r_1\in[\theta,0)$ with $r_2=\theta$ because, by \eqref{eq:c} and \eqref{eq:d}, we get $\beta(r_1,\theta)=\n$. However, sender~1 can profitably deviate from the prescribed strategy by delivering $r_1'=s(\theta)$, where $0<s(\theta)<\bar r_1(\theta)$ (Lemmata~\ref{lemma:swing} and \ref{lemma:cutoffs}), contradicting that $S_1(\theta)=\{r_1\}$. Consider next the case where $r_1\geq  0$ and $r_1\geq \theta$. If $s(r_1)\leq\ubar r_2(\theta)$, then it must be that $S_2(\theta)=\{\theta\}$. By Definition~\ref{def:cutoffs} and Lemma~\ref{lemma:swing} we have that $\ubar r_2(\theta)<0$ and $r_1\geq s(\ubar r_2(\theta))>0$. Since $\ubar r_2(\theta)<\theta$, sender~1 can profitably deviate from the prescribed strategy by reporting either $r_1'=s(\theta)\in(0,r_1)$ if $\theta<0$, or $r_1'=\theta$ if $\theta\geq 0$, as in both cases we get that $\beta(r_1',\theta)=\p$ and $C_1(r_1',\theta)<C_1(r_1,\theta)$. If instead $s(r_1)>\ubar r_2(\theta)$, then sender~2 must be delivering some report $r_2'\in(\ubar r_2(\theta),s(r_1))$. Therefore, if $r_1>\theta$, then sender~1 is strictly better off reporting $\theta$ rather than $r_1$ because $\beta(\theta,r_2')=\beta(r_1,r_2')=\n$ and $C_1(r_1,\theta)>0=C_1(\theta,\theta)$. If instead $r_1=\theta$, then $\theta\geq 0$ and since $\ubar r_2(\theta)\geq \ubar r_2(0)$ we have that $s(r_2')\leq \bar r_1(\theta)$ (Lemma~\ref{lemma:swing}). In this case, sender~1 can profitably deviate from the prescribed strategy by reporting $r_1'=s(r_2')$. Similar arguments apply to $S_2(\theta)=\{r_2\}$, completing the proof.
\end{proof}

\begin{Alemma}\label{lemma:support1}
	In an adversarial equilibrium, for every $\theta\in(\theta_l,\theta_h)$, supports $S_j(\theta)$ are such that
	\[
	\max S_1(\theta)\leq \min\left\{\bar r_1(0),\bar r_1(\theta),s\left(\ubar r_2(\theta)\right)\right\},
	\] 
	\[
	\min S_2(\theta)\geq \max\left\{\ubar r_2(0), \ubar r_2(\theta),s\left(\bar r_1(\theta)\right)\right\}.
	\] 
\end{Alemma}
\begin{proof}
	Consider an AE and a $\theta\in (\theta_l,\theta_h)$. By the definition of reach (equations~\eqref{eq:reach1} and \eqref{eq:reach2}) every $r_1>\bar r_1(\theta)$ is strictly dominated by truthful reporting, and thus $\max S_1(\theta)\leq \bar r_1(\theta)$. Similarly, we obtain that $\min S_2(\theta)\geq \ubar r_2(\theta)$ and therefore by \eqref{eq:d} and by Definition~\ref{def:swing} every $r_1>s(\ubar r_2(\theta))$ is dominated by $r_1'=s(\ubar r_2(\theta))$ and every $r_2<s(\bar r_1(\theta))$ is dominated by $r_2'=s(\bar r_1(\theta))$. Thus, $\max S_1(\theta) \leq s(\ubar r_2(\theta))$ and $\min S_2(\theta)\geq s(\bar r_1(\theta))$. For every $\theta\in[0,\theta_h)$ we have $\bar r_1(\theta)\geq \bar r_1(0)$ and $\ubar r_2(\theta)\geq \ubar r_2(0)$, and therefore $\min S_2(\theta)\geq \ubar r_2(0)$. Since $s(\ubar r_2(0))=\bar r_1(0)$ (Lemma~\ref{lemma:swing}), it follows from \eqref{eq:d} and Definition~\ref{def:swing} that $s(r_2)\leq \bar r_1(0)$ for every $r_2\in S_2(\theta)$, and therefore $\max S_1(\theta)\leq \bar r_1(0)$. Similarly, we obtain that $\min S_2(\theta)\geq s(\bar r_1(0))$ for every $\theta\in(\theta_l,0)$.
\end{proof}

\begin{Alemma}\label{lemma:support2}
	In an adversarial equilibrium, $r_2\notin S_2(\theta)$ for every $r_2\in(s(\min S_1(\theta)),\theta)$ and $\theta>0$, and $r_1\notin S_1(\theta)$ for every $r_1\in(\theta, s(\max S_2(\theta)))$ and $\theta<0$.
\end{Alemma}
\begin{proof}
	Consider $\theta\in(0,\theta_h)$. By Lemmata~\ref{lemma:monot} and \ref{lemma:swing} we have that $s(\min S_1(\theta))<0$, and by Definition~\ref{def:swing} we have that $\beta(r_1,r_2)=\p$ for every $r_1\in S_1(\theta)$ and $r_2\in(s\left(\min S_1(\theta)\right),\theta)$. Therefore, for sender~2 every $r_2\in(s\left(\min S_1(\theta)\right),\theta)$ is strictly dominated by truthful reporting, and therefore $r_2\notin S_2(\theta)$.  A similar argument applies to sender~1 for $\theta\in(\theta_l,0)$ and Lemma~\ref{lemma:supportscutoffs} proves the case of $\theta\notin (\theta_l,\theta_h)$, completing the proof.
\end{proof}

\begin{Alemma}\label{lemma:swingsupp}
	In an adversarial equilibrium, for every $\theta\in(\theta_l,\theta_h)$,
	\begin{itemize}[noitemsep]
		\item reports $r_1\in(\min S_1(\theta),\max S_1(\theta))$ have $s(r_1)>\ubar r_2(\theta)$;
		\item reports $r_2\in(\min S_2(\theta),\max S_2(\theta))$ have $s(r_2)<\bar r_1(\theta)$.
	\end{itemize}  
\end{Alemma}
\begin{proof}
	Suppose not, and consider $r_1'\in (\min S_1(\theta),\max S_1(\theta))$ for some $\theta\in(\theta_l,\theta_h)$ such that  $s(r_1')<\ubar r_2(\theta)$. By Definition~\ref{def:cutoffs} we have that $\ubar r_2(\theta)<0$ and by Lemma~\ref{lemma:swing} we have that $s(\ubar r_2(\theta))<r_1'$. This is in contradiction to Lemma~\ref{lemma:support1}, which states that $\max S_1(\theta)\leq s(\ubar r_2(\theta))$.  A similar argument holds for reports $r_2\in(\min S_2(\theta),\max S_2(\theta))$, completing the proof.
\end{proof}

\begin{Alemma}\label{lemma:noatom}
	In an adversarial equilibrium, $\alpha_j(r_j,\theta)=0$ for all $r_j\in (\min S_j(\theta),\max S_j(\theta))$, $j\in\{1,2\}$, and $\theta\in(\theta_l,\theta_h)$.
\end{Alemma}
\begin{proof}
	Consider $\theta\in(\theta_l,\theta_h)$ and suppose that there is an AE $\omega$ where sender~1's strategy $\phi_1(\theta)$ has an atom $\alpha_1(r_1',\theta)>0$ in some report $r_1'\in (\min S_1(\theta),\max S_1(\theta))$. By Lemma~\ref{lemma:swingsupp} we have that $s(r_1')>\ubar r_2(\theta)$. The expected payoff of sender~2, $W_2^\omega(\cdot,\theta)$, is discontinuous around $r_2=s(r_1')$ and therefore it must be that, for some $\epsilon>0$ small enough, $(s(r_1'),s(r_1')+\epsilon) \cap S_2(\theta)=\varnothing$. Therefore, there exists an $\epsilon'>0$ small enough such that $W_1^\omega(r_1'',\theta)>W_1^\omega(r_1',\theta)$ for some $r_1''\in\left(s\left(s(r_1')+\epsilon'\right),r_1'\right)$, where by Lemma~\ref{lemma:swing} we have that $s\left(s(r_1')+\epsilon'\right)<r_1'$, thus contradicting that this is an equilibrium. A similar argument applies to atoms in sender~2's strategy, completing the proof.	
\end{proof}

\begin{Alemma}\label{lemma:supp3}
	In an adversarial equilibrium, $\min S_1(\theta)=\theta$ for all $\theta\geq 0$, and $\max S_2(\theta)=\theta$ for all $\theta\leq 0$.
\end{Alemma}
\begin{proof}
	Consider an AE $\omega$ and $\theta\geq 0$. By Lemma~\ref{lemma:monot}, it must be that $\min S_1(\theta)\geq \theta$. Suppose by way of contradiction that $\min S_1(\theta)>\theta$. By Lemma~\ref{lemma:supportscutoffs} it has to be that $\theta<\theta_h$ and by Lemma~\ref{lemma:support2} we obtain that $S_2(\theta) \cap (s(\min S_1(\theta)),\theta)=\varnothing$. Therefore, unless sender~2's strategy has an atom $\alpha_2(s(\min S_1(\theta)),\theta)>0$, we have that $\Phi_2(s(\min S_1(\theta)),\theta)=\Phi_2(s(\theta),\theta)$. However, since $\beta(r_1,s(\min S_1(\theta)))=\p$ for all $r_1\in S_1(\theta)$ and $C_2(s(\min S_1(\theta)),\theta)>0$, it must be that $\alpha_2(s(\min S_1(\theta)),\theta)=0$ as $s(\min S_1(\theta))$ is strictly dominated by $r_2=\theta$. Hence we have that, for some $\epsilon>0$, $W_1^\omega(\theta,\theta)>W_1^\omega(r_1,\theta)$ for every $r_1\in[\min S_1(\theta), \min S_1(\theta)+\epsilon)\cap S_1(\theta)$, contradicting that there can be an AE with $\min S_1(\theta)>\theta$ for a $\theta\geq 0$. A similar argument holds for sender~2 and $\theta\leq 0$, completing the proof.
\end{proof}

\begin{Alemma}\label{lemma:supp4}
	In an adversarial equilibrium, $S_j(\theta)\setminus\{\theta\}$ contains more than one report for $j\in\{1,2\}$ and every $\theta\in(\theta_l,\theta_h)$.
\end{Alemma}
\begin{proof}
	Consider an AE $\omega$ and a state $\theta\in[0,\theta_h)$. By Lemma~\ref{lemma:supp3} we have that $\min S_1(\theta)=\theta$, and by Lemma~\ref{lemma:supportscutoffs} we have that $S_1(\theta)$ contains more than one report. Suppose by way of contradiction that $S_1(\theta)\setminus\{\theta\}=\{r_1\}$ for some $r_1>0$. Since $C_1(r_1,\theta)>0$, it must be that, in equilibrium, $r_1$ induces $\p$ with strictly higher probability than truthful reporting. This implies that there is some $r_2\in \left[s(r_1),s(\theta)\right)$ in the support of sender~2's strategy, $r_2\in S_2(\theta)$. Since reports that are further away from the realized state are more costly, it must be that $\alpha_2(r_2',\theta)>0$ for some $r_2'\in \left[s(r_1),s(\theta)\right)$, and $\phi_2(r_2,\theta)=0$ for all $r_2\in \left[s(r_1),r_2'\right)$. But then $W_1^\omega(s(r_2'),\theta)>W_1^\omega(r_1,\theta)$, contradicting that this is an equilibrium. 
	
	Consider now the case where $\theta\in(\theta_l,0)$ and suppose again that $S_1(\theta)\setminus\{\theta\}=\{r_1\}$. By Lemma~\ref{lemma:supportscutoffs}, we have $S_j(\theta)$ contains more than one report for $j\in\{1,2\}$, and therefore $\min S_1(\theta)=\theta$. By Lemmata~\ref{lemma:support1} and \ref{lemma:support2} we have that $r_1\geq s(\theta)>0$ and $\max S_2(\theta)=\theta$. If $r_1=s(\theta)$, then sender~2 can profitably deviate from the prescribed strategy by always reporting $\theta-\epsilon$ for some $\epsilon>0$ small enough. If instead $r_1>s(\theta)$, then it must be that $S_2(\theta)\cap [s(r_1),\theta)=\varnothing$ as every $r_2\in [s(r_1),\theta)$ would be strictly dominated by truthful reporting. Since $S_2(\theta)$ contains more than one report, there must be some $r_2<s(r_1)$ such that $r_2\in S_2(\theta)$. Therefore, $W_1^\omega(s(\theta),\theta)>W_1^\omega(r_1,\theta)$, contradicting that this is an equilibrium. A similar argument applies to $S_2(\theta)\setminus\{\theta\}$, completing the proof.
\end{proof}


\begin{Alemma}\label{lemma:convex}
	In an adversarial equilibrium,  $S_j(\theta)\setminus\{\theta\}$ is convex for all $\theta\in(\theta_l,\theta_h)$ and $j\in\{1,2\}$.
\end{Alemma} 

\begin{proof}
	Consider an AE $\omega$ and a state $\theta \in (\theta_l,\theta_h)$. By Lemma~\ref{lemma:supp4} we have that $S_j(\theta)\setminus\{\theta\}$ contains more than one report, $j\in\{1,2\}$. Suppose by way of contradiction that $S_1(\theta)\setminus\{\theta\}$ is not convex, but instead there are two reports $r_1',r_1''\in S_1(\theta)\setminus\{\theta\}$ with $r_1'<r_1''$, such that $r_1\notin S_1(\theta)\setminus\{\theta\}$ for every $r_1\in (r_1',r_1'')$. By Lemmata~\ref{lemma:monot}, \ref{lemma:swing}, and \ref{lemma:support2} we have that $r_1'>0$, $r_1'\geq s(\theta)$, and $s(r_1'')<s(r_1')<0$. Since $C_1(r_1'',\theta)>C_1(r_1',\theta)$ and $\frac{d C_j(r,\theta)}{d r}>0$ for every $r>\theta$, it must be that every report $r_1\geq r_1''$ such that $\phi_1(r_1,\theta)>0$ induces the implementation of alternative $\p$ with strictly higher probability than every report $r_1'''\leq r_1'$ such that $\phi_1(r_1''',\theta)>0$. This is possible only if $r_2\in S_2(\theta)$ for some $r_2\in[s(r_1''),s(r_1')]$. Since $\Phi_1(r_1,\theta)$ is constant for all $r_1\in(r_1',r_1'')$, it must be that sender~2's strategy has an atom $\alpha_2(r_2,\theta)>0$ in some $r_2\in(s(r_1''),s(r_1')]$, and $\phi_2(r_2',\theta)=0$ for all $r_2'\in[s(r_1''),s(r_1')]$ such that $r_2'\neq r_2$. However, for some $\epsilon>0$ small enough we have that $W_1^\omega(s(r_2),\theta)>W_1^\omega(r_1,\theta)$ for all $r_1\in[r_1'',r_1''+\epsilon)$ such that $r_1\in S_1(\theta)$, where $s(r_2)<r_1''$, contradicting that this is an equilibrium. A similar argument applies to $S_2(\theta)\setminus\{\theta\}$, completing the proof.
\end{proof}

\begin{Alemma}\label{lemma:noatom2}
	In an adversarial equilibrium, the strategies $\phi_j(\theta)$ have no atoms in $S_j(\theta)\setminus\{\theta\}$ for every $\theta\in(\theta_l,\theta_h)$ and $j\in\{1,2\}$.
\end{Alemma}

\begin{proof}
	Consider an AE $\omega$. Lemma~\ref{lemma:supportscutoffs} shows that $S_j(\theta)$ contains more than one report for all $\theta\in(\theta_l,\theta_h)$ and Lemma~\ref{lemma:noatom} shows that $\phi_j(\theta)$ has no atoms in $(\min  S_j(\theta),\max S_j(\theta))$. Consider a state $\theta\in(\theta_l,\theta_h)$, and suppose that $\phi_1(\theta)$ has an atom in $\max S_1(\theta)$, i.e., $\alpha_1(\max S_1(\theta),\theta)>0$. By Lemma~\ref{lemma:support1}, we have that $\max S_1(\theta)\leq  \min\{s(\ubar r_2(\theta)),\bar r_1(\theta)\}$ and $\min S_2(\theta)\geq\max\{\ubar r_2(\theta),s(\bar r_1(\theta))\}$. If $\min S_2(\theta)>s(\max S_1(\theta))$, then $W_1^\omega(r_1,\theta)>W_1^\omega(\max S_1(\theta),\theta)$ for any $r_1\in[s(\min S_2(\theta)),\max S_1(\theta))$. If $\min S_2(\theta)=s(\max S_1(\theta))$, then, since sender~2's expected payoff $W_2^\omega(\cdot,\theta)$ is discontinuous at $r_2=s(\max S_1(\theta))$, it must be that $r_2'\notin S_2(\theta)$ for all $r_2' \in (s(\max S_1(\theta)),s(\max S_1(\theta))+\epsilon)$ and some small $\epsilon>0$. Otherwise, for some $\epsilon'>0$, $W_2^\omega(s(\max S_2(\theta))-\epsilon',\theta)>W_2^\omega(r_2',\theta)$ for any $r_2' \in [s(\max S_1(\theta)),s(\max S_1(\theta))+\epsilon]$. However, this would contradict either Lemma~\ref{lemma:supp4} or Lemma~\ref{lemma:convex}, and therefore it would not be possible in an AE.
	
	Suppose now that $\phi_1(\theta)$ has an atom in $\min S_1(\theta)$, i.e., $\alpha_1(\min S_1(\theta),\theta)>0$. By Lemma~\ref{lemma:supp3}, if $\theta\geq 0$ then $\min S_1(\theta)=\theta$, and therefore let us suppose that $\theta\in(\theta_l,0)$ and that $\min S_1(\theta)>\theta$ when $\theta<0$. By Lemmata~\ref{lemma:swing}, \ref{lemma:support2}, and \ref{lemma:supp3} we have that $\min S_1(\theta)\geq s(\theta)>0$. If $\min S_1(\theta)=s(\theta)$, then it must be that $\phi_2(\theta,\theta)=0$, as $W_2^\omega(\theta-\epsilon,\theta)>W_2^\omega(\theta,\theta)$ for some $\epsilon>0$ small enough. But then the atom in $\min S_1(\theta)$ would be strictly dominated by truthful reporting as $C_1(s(\theta),\theta)>0$ and $\beta(s(\theta),r_2)=\n$ for every $r_2\in S_2(\theta)$, contradicting that this is an equilibrium. Consider now the case where $\min S_1(\theta)>s(\theta)$. We have that $\Phi_1(r_1,\theta)=0$ for every $r_1<\min S_1(\theta)$, and by Lemma~\ref{lemma:swing} we have that $s(\min S_1(\theta))<\theta$. Therefore, it must be that $\phi_2(r_2,\theta)=0$ for every $r_2\in[s(\min S_1(\theta)),\theta)$. However this implies that, for sender~1, $\min S_1(\theta)$ is dominated by $s(\theta)$, contradicting that this can be an equilibrium. Similar arguments hold for atoms $\alpha_2(r_2,\theta)$ for $r_2\in S_2(\theta)\setminus\{\theta\}$, completing the proof.
\end{proof}


\begin{Alemma}\label{lemma:convex2}
	In an adversarial equilibrium, $S_1(\theta)$ is convex for all $\theta\geq 0$ and $S_2(\theta)$ is convex for all $\theta \leq 0$.
\end{Alemma}
\begin{proof}
	Consider an AE $\omega$ and suppose by way of contradiction that $S_1(\theta)$ is not convex for some $\theta\in[0,\theta_h)$. By Lemma~\ref{lemma:supp3} we have that $\min S_1(\theta)=\theta$, and by Lemma~\ref{lemma:convex} we have that $S_1(\theta)\setminus\{\theta\}$ is convex. Therefore, it must be that $\min S_1(\theta)\setminus\{\theta\}>\theta$ and $\phi_1(r_1,\theta)=0$ for every $r_1\in\left(\theta,\min S_1(\theta)\setminus\{\theta\}\right)$. In equilibrium, every $r_1>\min S_1(\theta)\setminus\{\theta\}$ such that $\phi_1(r_1,\theta)>0$ must yield the implementation of alternative $\p$ with strictly higher probability than truthful reporting, as $C_1(r_1,\theta)>0$. This is possible only if $\phi_2(r_2,\theta)>0$ for some $r_2\in \left[s\left(\min S_1(\theta)\setminus\{\theta\}\right),s(\theta)\right)$. However, for some $\epsilon>0$ small enough, it must be that $\phi_2(r_2',\theta)=0$ for every $r_2'\in\left[s\left(\min S_1(\theta)\setminus\{\theta\}\right),s(\theta)-\epsilon\right)$, as every such report $r_2'$ is dominated by reporting $s(\theta)-\epsilon$. Therefore, there exists an $\epsilon'>0$ such that $W_1^\omega(s(s(\theta)-\epsilon),\theta)>W_1^\omega(r_1',\theta)$ for all $r_1'\in [\min S_1(\theta)\setminus\{\theta\},\min S_1(\theta)\setminus\{\theta\}+\epsilon')$, contradicting that this is an equilibrium. Lemma~\ref{lemma:supportscutoffs} considers the case where $\theta\notin (\theta_l,\theta_h)$, and a similar argument applies to states $\theta\leq 0$ and support $S_2(\theta)$.
\end{proof}

\subsubsection{Proof of Proposition~\ref{prop:adveqm}}\label{app:aestrat}

\begin{step}\label{step:strategies}
	In an adversarial equilibrium, for every $\theta\in(\theta_l,\theta_h)$ and $i,j\in\{1,2\}$ with $i\neq j$, sender~$j$ delivers report $r_j\in S_j(\theta)\setminus\{\theta\}$ according to
	\[
	\psi_j(r_j,\theta)=\frac{k_i}{-u_i(\theta)}\frac{d C_i(s(r_j),\theta)}{d r_j}.
	\]
\end{step}

\begin{proof}
	Consider an AE and a state $\theta\in (\theta_l,\theta_h)$. Given strategy $\phi_1(\theta)$, sender~2 gets an expected utility of $W^\omega_2(r_2,\theta)=(1-\Phi_1(s(r_2),\theta))u_2(\theta)-k_2C_2(r_2,\theta)$ from delivering $r_2 \in S_2(\theta)\setminus\{\theta\}$. By Lemmata~\ref{lemma:convex} and \ref{lemma:noatom2} we have that $S_j(\theta)\setminus\{\theta\}$ is convex and atomless. By Lemmata~\ref{lemma:monot}, \ref{lemma:cutoffs}, and \ref{lemma:support1}, we have that $S_j(\theta)\subset \left[\ubar r_2(0),\bar r_1(0)\right]$ for all $\theta\in(\theta_l,\theta_h)$, and therefore by Lemma~\ref{lemma:swing} we have that $\frac{d s(r)}{dr}<0$ for all  $r_j\in S_j(\theta)$. Therefore, we can set $\frac{d W^\omega_2(r_2,\theta)}{d r_2}\big|_{r_2\in S_2(\theta)\setminus\{\theta\}}=0$, and since $\phi_j(r_j,\theta)=\psi_j(r_j,\theta)$ for all $r_j\in S_j(\theta)\setminus\{\theta\}$ (Lemma~\ref{lemma:noatom2}), we obtain the partial pdf 
	\[
	\psi_1(s(r_2),\theta)=\frac{k_2}{-u_2(\theta)}\frac{d C_2(r_2,\theta)}{d r_2}\frac{d r_2}{d s(r_2)}=\frac{k_2}{-u_2(\theta)}\frac{d C_2(r_2,\theta)}{d s(r_2)}.
	\] 
	By replacing $r_1=s(r_2)$ we obtain that $\psi_1(r_1,\theta)=\frac{k_2}{-u_2(\theta)}\frac{d C_2(s(r_1),\theta)}{d r_1}$ for $r_1\in S_1(\theta)\setminus\{\theta\}$. Similarly, we obtain that for $r_2 \in S_2(\theta)\setminus\{\theta\}$, $\psi_2(r_2,\theta)=\frac{k_1}{-u_1(\theta)}\frac{d C_1(s(r_2),\theta)}{d r_2}$.
\end{proof}


\begin{step}\label{step:supports}
	In an adversarial equilibrium, for every state $\theta\in (\theta_l,\theta_h)$, supports $S_j(\theta)$ are
	\begin{displaymath}
		S_1(\theta)=\{\theta\}\cup \left[\max \left\{s(\theta),\theta\right\},\min \left\{\bar r_1(\theta), s\left(\ubar r_2(\theta)\right)\right\}\right],
	\end{displaymath}
	\begin{displaymath}
		S_2(\theta)= \{\theta\}\cup\left[\max \left\{\ubar r_2(\theta), s\left(\bar r_1(\theta)\right)\right\},\min \left\{s(\theta),\theta\right\}\right].
	\end{displaymath}
\end{step}

\begin{proof}
	Consider an adversarial equilibrium and a state $\theta\in[0,\theta_h)$. Since for every $\theta\geq 0$ we have that $\theta\in S_1(\theta)$ (Lemma~\ref{lemma:supp3}) and both sets $S_1(\theta)$ and $S_1(\theta)\setminus\{\theta\}$ are convex (Lemmata~\ref{lemma:convex2} and \ref{lemma:convex}), it follows that $S_1(\theta)=[\theta,\max S_1(\theta)]$.
	
	Lemma~\ref{lemma:convex} shows that also $S_2(\theta)\setminus\{\theta\}$ is convex. Since $\min S_1(\theta)=\theta$, Lemma~\ref{lemma:support2} says that when $\theta>0$ we have that $\phi_2(r_2,\theta)=0$ for all $r_2\in(s(\theta),\theta)$, and therefore $\max S_2(\theta)\setminus\{\theta\}\leq s(\theta)$ for all $\theta\in(0,\theta_h)$. Suppose that $\max S_2(\theta) \setminus \{\theta\}<s(\theta)$. In this case, it must be that $\phi_1(r_1,\theta)=0$ for every $r_1\in(\theta, s(\max S_2(\theta) \setminus \{\theta\}))$, as for sender~1 every such a report $r_1$ would be dominated by truthful reporting. This is in contradiction to Lemma~\ref{lemma:convex2}, and therefore it must be that $\max S_2(\theta) \setminus \{\theta\}=s(\theta)$ for every $\theta\in(0,\theta_h)$. When $\theta=0$, we have that $\max S_2(0)=0$ (Lemma~\ref{lemma:supp3}). 
	
	Lemma~\ref{lemma:noatom2} shows that $\phi_2(r_2,\theta)$ is atomless in $S_2(\theta)\setminus\{\theta\}$. Therefore, for $r_2\in S_2(\theta)\setminus\{\theta\}$ we have that $\Phi_2(r_2,\theta)=\Psi_2(r_2,\theta)$, and therefore by using Step~\ref{step:strategies} we can write
	\[
	\Phi_2(r_2,\theta)|_{r_2\in S_2(\theta)\setminus\{\theta\}}=\int_{\min S_2(\theta)}^{r_2} \psi_2(r,\theta)d r =\frac{k_1}{u_1(\theta)}\left[ C_1(s(\min S_2(\theta)),\theta) - C_1(s(r_2),\theta) \right].
	\]
	The probability that sender~2 misreports information in state $\theta\in(0,\theta_h)$ is therefore
	\begin{equation}\label{eq:problie}
		\Phi_2(s(\theta),\theta)=\frac{k_1}{u_1(\theta)} C_1(s(\min S_2(\theta)),\theta).
	\end{equation}
	Since $\min S_2(\theta)\geq \ubar r_2(\theta)$ (Lemma~\ref{lemma:support1}), it follows from Lemma~\ref{lemma:swing} that, for every $\theta\in (0,\theta_h)$, $s(\min S_2(\theta))<\bar r_1(\theta)$. Lemma~\ref{lemma:supp4} shows that the set $S_2(\theta)\setminus\{\theta\}$ is not a singleton, and since $\max S_2(\theta)\setminus\{\theta\}\leq s(\theta)$ it must be that $\min S_2(\theta)<s(\theta)$. Therefore by Lemma~\ref{lemma:swing} we have that $s(\min S_2(\theta))\in(\theta,\bar r_1(\theta))$ for $\theta\in(0,\theta_h)$. Finally, by the definition of reach we get that $C_1(\bar r_1(\theta),\theta)=u_1(\theta)/k_1$, and $C_1(r_1,\theta)<u_1(\theta)/k_1$ for every $r_1\in[\theta,\bar r_1(\theta))$. Therefore, it follows that $\Phi_2(s(\theta),\theta)\in(0,1)$ for every $\theta\in (0,\theta_h)$. By using $s(s(r))=r$ and $s(0)=0$ (Lemma~\ref{lemma:swing}), when $\theta=0$ we obtain that $\Phi_2(s(0),0)=1$ only if $\min S_2(0)=\ubar r_2(0)$.
	
	The above argument shows that $\theta\in S_2(\theta)$ and that $\phi_2(\theta)$ has an atom in $r_2=\theta$ of size $\alpha_2(\theta)=1-\Phi_2(s(\theta),\theta)$.	Lemma~\ref{lemma:monot} implies that every pair of on path reports $(r_1,r_2)$ such that $r_j\geq 0$, $j\in\{1,2\}$, must yield $\beta(r_1,r_2)=\p$. Therefore, by reporting truthfully when $\theta\geq 0$, sender~2 obtains a payoff of $W^\omega_2(\theta,\theta)=u_2(\theta)$. It must be that  $\max S_1(\theta)\leq s(\min S_2(\theta))$, as otherwise every report $r_1>s(\min S_2(\theta))$ would be dominated by $s(\min S_2(\theta))$. Since $\phi_1(\theta)$ has no atom in $s(\min S_2(\theta))>\theta$ (Lemma~\ref{lemma:noatom2}), by reporting $r_2=\min S_2(\theta)$ sender~2 (almost) always induces the selection of his preferred alternative $\n$, and gets an expected payoff of $W^\omega_2(\min S_2(\theta),\theta)=-k_2C_2(\min S_2(\theta),\theta)$.

	In equilibrium, each sender must receive the same expected payoff from delivering any report that is in the support of its own strategy. Since by the definition of reach we obtain $C_2(\ubar r_2(\theta),\theta)=-u_2(\theta)/k_2$, it follows that $W^\omega_2(\min S_2(\theta),\theta)=u_2(\theta)=W^\omega_2(\theta,\theta)$ only if $\min S_2(\theta)=\ubar r_2(\theta)$. Therefore, for a $\theta\in[0,\theta_h)$, we have that $S_2(\theta)=[\ubar r_2(\theta),s(\theta)]\cup\{\theta\}$. It also follows that $\max S_1(\theta)=s(\ubar r_2(\theta))$: if $\max S_1(\theta)<s(\ubar r_2(\theta))$, then $\ubar r_2(\theta)<s(\max S_1(\theta))$ and every $r_2<s(\max S_1(\theta))$ would be strictly dominated by $s(\max S_1(\theta))$. Thus, $S_1(\theta)=[\theta,s(\ubar r_2(\theta))]$. Similar arguments apply to the case where $\theta\in(\theta_l,0)$, completing the proof.
\end{proof}

\begin{step}\label{step:atoms}
	In an adversarial equilibrium, for every state $\theta\in (\theta_l,\theta_h)$, strategies $\phi_j(\theta)$ have an atom at $r_j=\theta$ of size $\alpha_j(\theta)$, where
	\begin{displaymath}
		\alpha_1(\theta)=
		\begin{cases}
			\frac{k_2}{-u_2(\theta)} C_2\left( s(\theta),\theta \right) & \quad \text{if } \; \theta\in [0,\theta_h)\\ 
			1-\frac{k_2}{-u_2(\theta)}C_2\left(s(\bar r_1(\theta)),\theta\right) & \quad \text{if } \theta\in(\theta_l,0],
		\end{cases}
	\end{displaymath}
	\begin{displaymath}
		\alpha_2(\theta)=
		\begin{cases}
			1-\frac{k_1}{u_1(\theta)}C_1\left(s(\ubar r_2(\theta)),\theta\right) & \quad \text{if } \; \theta\in [0,\theta_h)\\ 
			\frac{k_1}{u_1(\theta)}C_1\left(s(\theta),\theta\right) & \quad \text{if } \theta\in(\theta_l,0].
		\end{cases}
	\end{displaymath}
\end{step}

\begin{proof}
	Consider an adversarial equilibrium and a state $\theta\in[0,\theta_h)$. The proof of Step~\ref{step:supports} shows that $\phi_2(\theta)$ has an atom in $r_2=\theta$ of size $\alpha_2(\theta)=1-\Phi_2(s(\theta),\theta)$. From equation~\eqref{eq:problie} and given $\min S_2(\theta)=\ubar r_2(\theta)$, we obtain that
	\[
	\alpha_2(\theta)=1-\frac{k_1}{u_1(\theta)} C_1(s(\ubar r_2(\theta)),\theta).
	\]
	By Lemma~\ref{lemma:noatom2}, sender~1's strategy $\phi_1(\theta)$ admits an atom only in $\min S_1(\theta)=\theta$. Therefore, we can use Step~\ref{step:strategies} to write
	\begin{displaymath}
		\begin{split}
			\Phi_1(r_1,\theta)|_{r_1\in S_1(\theta)} & =\alpha_1(\theta)+\int_{\theta}^{r_1}\psi_1(r,\theta)dr \\
			&= \alpha_1(\theta)+\frac{k_2}{-u_2(\theta)}\left[ C_2(s(r_1),\theta)-C_2(s(\theta),\theta) \right].
		\end{split}
	\end{displaymath}
	Since $\max S_1(\theta)=s(\ubar r_2(\theta))$, it must be that $\Phi_1(s(\ubar r_2(\theta)),\theta)=1$. By using $s(s(\ubar r_2(\theta)))=\ubar r_2(\theta)$ (Lemma~\ref{lemma:swing}) and given that from the definition of reach we obtain $C_2(\ubar r_2(\theta),\theta)=-u_2(\theta)/k_2$, we have that
	\begin{displaymath}
		\begin{split}
			\Phi_1(s(\ubar r_2(\theta)),\theta) & =\alpha_1(\theta)+\frac{k_2}{-u_2(\theta)}\left[ C_2(s(s(\ubar r_2(\theta))),\theta)-C_2(s(\theta),\theta)  \right] \\
			&=\alpha_1(\theta)+1-\frac{k_2}{-u_2(\theta)}C_2(s(\theta),\theta)=1,
		\end{split}
	\end{displaymath}
	from which we obtain that
	\[
	\alpha_1(\theta)=\frac{k_2}{-u_2(\theta)}C_2(s(\theta),\theta).
	\]
	A similar procedure can be used for $\theta\in(\theta_l,0)$, completing the proof.
\end{proof}


\begin{Alemma}\label{lemma:interval}
	In an adversarial equilibrium, for every (on path) pair of reports $(r_1,r_2)$ such that $r_2=s(r_1)$, the decision maker's posterior beliefs are
	\[
	p(\theta\,|\,r_1,r_2)>0 \;\text{ if and only if }\; \theta\in[\max\{r_2,\bar r_1^{-1}(r_1)\},\min\{r_1,\ubar r_2^{-1}(r_2)\}].
	\] 
\end{Alemma}
\begin{proof}
	Consider an AE and a pair of reports $(r_1,r_2)$ such that $\bar r_1(0)\geq r_1>0>r_2\geq\ubar r_2(0)$. Given equilibrium supports in Step~\ref{step:supports}, all such pairs are on path (e.g., for $\theta=0$). Upon observing $(r_1,r_2)$, the decision maker forms posterior beliefs $p(\theta\,|\,r_1,r_2)$. By Lemma~\ref{lemma:monot}, it must be that $p(\theta\,|\,r_1,r_2)=0$ for every $\theta\notin[r_2,r_1]$. By Lemma~\ref{lemma:supportscutoffs}, it must be that $p(\theta\,|\,r_1,r_2)=0$ for every $\theta\notin[\theta_l,\theta_h]$. By Step~\ref{step:supports} we have that $\min S_2(\theta)\geq\ubar r_2(\theta)$ and $\max S_1(\theta)\leq \bar r_1(\theta)$, and therefore $p(\theta\,|\,r_1,r_2)=0$ for every $\theta\notin\left[\bar r_1^{-1}(r_1), \ubar r_2^{-1}(r_2)\right]$, where from equations~\eqref{eq:reach1} and \eqref{eq:reach2} we obtain that
	\[
	\bar r_1^{-1}(r_1)=\min\left\{\theta\in\Theta \,|\, u_1(\theta)=k_1C_1(r_1,\theta)  \right\},
	\]
	\[
	\ubar r_2^{-1}(r_2)=\max\left\{\theta\in\Theta \,|\, -u_2(\theta)=k_2C_2(r_2,\theta)  \right\}.
	\]
	From Step~\ref{step:supports} we also have that, for every $\theta\in[0,\theta_h)$, $\max S_1(\theta)=s(\ubar r_2(\theta))\leq r_1(\theta)$. Therefore, given the report $r_1\in(0,\bar r_1(0)]$, it must be that $p(\theta\,|\,r_1,r_2)=0$ for all $\theta$ such that $s(\ubar r_2(\theta))< r_1$. By Lemma~\ref{lemma:swing} and since $d \ubar r_2(\theta)/d\theta>0$, there is a state $\theta'$ such that $s(\ubar r_2(\theta'))=r_1$. Denote such a state by $t_1(r_1):=\{\theta\in\Theta\,|\,s(\ubar r_2(\theta))=r_1\}$, where $t_1(r_1)>0$ and $d t_1(r_1)/dr_1>0$. Similarly, denote $t_2(r_2):=\{\theta\in\Theta\,|\,s(\bar r_1(\theta))=r_2\}$. Given equilibrium supports, it must be that $p(\theta\,|\,r_1,r_2)=0$ for all $\theta\notin[t_2(r_2),t_1(r_1)]$.
	
	By Lemma~\ref{lemma:swing} and since $s(\ubar r_2(\theta_h))=\theta_h$ (Definition~\ref{def:cutoffs}), we obtain that $t_1(r_1)\leq\theta_h$ for every $r_1\in[\theta_h,\bar r_1(0)]$, and therefore $\min\{r_1,t_1(r_1)\}\leq \theta_h$ for all $r_1\in(0,\bar r_1(0)]$. Similarly, we get that $\max\{r_2,t_2(r_2)\}\geq \theta_l$ for all $r_2\in[\ubar r_2(0),0)$. Therefore, we have that $p(\theta\,|\,r_1,r_2)=0$ for every $\theta\notin[\max\{r_2,\bar r_1^{-1}(r_1), t_2(r_2)\},\min\{r_1,\ubar r_2^{-1}(r_2),t_1(r_1)\}]$, and by Step~\ref{step:supports} we obtain that $p(\theta\,|\,r_1,r_2)\propto f(\theta) \cdot\phi_1(r_1,\theta)\cdot\phi_2(r_2,\theta)>0$ otherwise.
	
	Consider now the case where $r_2=s(r_1)$ (or, by Lemma~\ref{lemma:swing}, $r_1=s(r_2)$). By definition, in state $\theta'=t_1(r_1)$ we have that $s(\ubar r_2(\theta'))=r_1$. Therefore, we get that $s(r_1)=\ubar r_2(\theta')=r_2$ and $\ubar r_2^{-1}(r_2)=\theta'=t_1(r_1)$. Similarly, we obtain that $\bar r_1^{-1}(r_1)=t_2(r_2)$. Therefore, for every pair of reports $(r_1,s(r_1))$ we have that $p(\theta\,|\,r_1,s(r_1))>0$ if and only if $\theta\in\left[\max\left\{r_2,\bar r_1^{-1}(r_1)\right\},\min\left\{r_1,\ubar r_2^{-1}(r_2)\right\}\right]$.
\end{proof}


After establishing the senders' equilibrium supports and strategies, I can now proceed to study the decision maker's posterior beliefs. It is sufficient to examine how posterior beliefs $p$ shape the swing report function $s(r)$. By Lemma~\ref{lemma:swing}, we have that $s(r)\in \left[\ubar r_2(0),\bar r_1(0)\right]$ for every $r\in \left[\ubar r_2(0),\bar r_1(0)\right]$, with $s(r)< 0$ if $r>0$, $s(r)>0$ if $r<0$, and $s(0)=0$. Given the supports and the strategies as in Steps~\ref{step:strategies}, \ref{step:supports}, and \ref{step:atoms}, we obtain that every pair of reports $(r_1,r_2)$ such that $\ubar r_2(0) \leq r_2 <0 <r_1 \leq \bar r_1(0)$ is on path. By Definition~\ref{def:swing} and Lemma~\ref{lemma:swing} we have that, for a pair of reports $(r_1,r_2=s(r_1))$,
\[
U_\dm (r_1,s(r_1))= U_\dm (s(r_2),r_2) = \int_\Theta u_\dm(\theta)p(\theta\,|\,r_1,s(r_1))d\theta=0.
\]
Therefore, we can use $p(r_1,s(r_1)\,|\,\theta)=\phi_1(r_1,\theta)\cdot\phi_2(s(r_1),\theta)$ and previous results to show how posterior beliefs $p$ pin down the swing report function $s(r)$ in an adversarial equilibrium. The next proposition shows how the swing report function depends on the model's parameters. 


\begin{step}\label{step:swing}
	In an adversarial equilibrium, the swing report function $s(r_i)$ is implicitly defined for $i,j\in\{1,2\}$, $i\neq j$, and $r_i\in \left[\ubar r_2(0),\bar r_1(0)\right]$, as
	\begin{equation}\label{eq:swing}
		s(r_i)=\left\{r_j\in \Theta \; \bigg| \; \int_{\max\left\{r_2,\bar r_1^{-1}(r_1)\right\}}^{\min\left\{r_1,\ubar r_2^{-1}(r_2)\right\}} f(\theta)\frac{u_\dm(\theta)}{u_1(\theta) u_2(\theta)}\frac{d C_j(r_j,\theta)}{d r_j}\frac{d C_i(r_i,\theta)}{d r_i}d\theta=0\right\}.
	\end{equation}
\end{step}

\begin{proof}
	Given the equilibrium reporting strategies $\phi_j(r_j\,|\,\theta)=\delta(r_j-\theta)\alpha_j(\theta)+\psi_j(r_j\,|\,\theta)$, $j\in\{1,2\}$ (Steps~\ref{step:strategies}, \ref{step:supports}, and \ref{step:atoms}), the mixed probability distribution $p(r_1,r_2\,|\,\theta)=\phi_1(r_1,\theta)\phi_2(r_2,\theta)$ is
	\begin{displaymath}
		\begin{split}
			p(r_1,r_2\,|\,\theta) =\;&\;\delta(r_1-\theta)\delta(r_2-\theta)\alpha_1(\theta) \alpha_2(\theta)  + \delta(r_1-\theta)\alpha_1(\theta) \psi_2(r_2,\theta)\\
			&\;+\delta(r_2-\theta) \psi_1(r_1,\theta)\alpha_2(\theta) +  \psi_1(r_1,\theta)\psi_2(r_2,\theta).
		\end{split}
	\end{displaymath}
	Consider a pair of reports $(r_1,r_2)$ such that $\bar r_1(0)\geq r_1>0>r_2\geq \ubar r_2(0)$ and $r_2=s(r_1)$ (as by Lemma~\ref{lemma:swing} we have that if $r>0$, then $s(r)<0$). Since $C_i(\theta,\theta)=0$ for every $\theta\in\Theta$ and $i\in\{1,2\}$, we obtain that $\psi_j(s(\theta),\theta)=0$ for $j\in\{1,2\}$, and therefore $p(r_1,s(r_1)\,|\,\theta)=\psi_1(r_1,\theta)\psi_2(s(r_1),\theta)$.

	The swing report $s(r_1)$ is defined in Definition~\ref{def:swing} to be the report $r_2\in \Theta$ such that $U_\dm (r_1,r_2)=\int_\Theta u_\dm(\theta)p(\theta\,|\,r_1,r_2)d\theta=0$, and by Lemma~\ref{lemma:swing} we know that $s(r_1)\in[\ubar r_2(0),0)$. By Lemma~\ref{lemma:interval} we have that $p(\theta\,|\,r_1,s(r_1))>0$ if and only if 
	\[
	\theta\in[\max\{r_2,\bar r_1^{-1}(r_1)\},\min\{r_1,\ubar r_2^{-1}(r_2)\}].
	\] 
	Therefore by using Bayes' rule we can rewrite the condition $U_\dm(r_1,s(r_1))=0$ as $G_s(r_1,s(r_1))=0$, where
	\[
	G_s(r_1,r_2)=\frac{1}{p(r_1,r_2)}\int_{\max\left\{r_2,\bar r_1^{-1}(r_1)\right\}}^{\min\left\{r_1,\ubar r_2^{-1}(r_2)\right\}}u_\dm(\theta) f(\theta) \psi_1(r_1,\theta) \psi_2(r_2,\theta)d\theta.
	\]
	By substituting for the equilibrium strategies $\psi_j(r_j,\theta)$ as described in Step~\ref{step:strategies}, we obtain the implicit definition of the swing report given in equation~\eqref{eq:swing}.
\end{proof}

\begin{proof}[{\bf Proof of Proposition~\ref{prop:adveqm}.}]
	The proof follows directly from Steps~\ref{step:strategies} to \ref{step:swing}. Part i) follows from Lemma~\ref{lemma:supportscutoffs}. Steps~\ref{step:strategies}--\ref{step:atoms} provide a characterization of the senders' reporting strategies and their supports. The CDF $\Phi_j(r_j,\theta)$ is obtained by integrating $\phi_j(r_j,\theta)$ from $\min S_j(\theta)$ or $-\infty$ to $r_j$. Part iii) follows from Step~\ref{step:swing} and the observation that, by construction, posterior beliefs $p$ must satisfy conditions~\eqref{eq:d} and \eqref{eq:c}.
\end{proof}

\begin{Acorollary}\label{cor:prob}
	In an AE, the probability of full revelation and the probability that senders deliver the same report increase as the realized state is further away from zero.
\end{Acorollary}
\begin{proof}
	Given the reporting strategies as in Proposition~\ref{prop:adveqm} (and Steps~\ref{step:strategies}--\ref{step:swing}), the probability that senders deliver the same report in state $\theta$ is $\alpha_1(\theta)\alpha_2(\theta)$. The probability that the decision maker fully learns the state is $\alpha_1(\theta)$ when $\theta\geq 0$ and $\alpha_2(\theta)$ otherwise. The proof follows directly from the following two derivatives. 
	\begin{displaymath}
		\frac{d\alpha_1(\theta)}{d\theta}=
		\begin{cases}
			\frac{k_2}{u_2(\theta)^2}\frac{du_2(\theta)}{d\theta} C_2\left( s(\theta),\theta \right) + \frac{k_2}{-u_2(\theta)}\frac{dC_2\left( s(\theta),\theta \right)}{d\theta}>0 & \quad \text{if } \; \theta\in [0,\theta_h)\\ 
			-\frac{k_2}{u_2(\theta)^2}\frac{d u_2(\theta)}{d\theta}C_2\left(s(\bar r_1(\theta)),\theta\right) - \frac{k_2}{-u_2(\theta)}\frac{dC_2\left(s(\bar{r}_1(\theta)),\theta\right)}{d\theta}<0 & \quad \text{if } \theta\in(\theta_l,0),
		\end{cases}
	\end{displaymath}
	\begin{displaymath}
		\frac{d\alpha_2(\theta)}{d\theta}=
		\begin{cases}
			\frac{k_1}{u_1(\theta)^2}\frac{du_1(\theta)}{d\theta}C_1\left(s(\ubar r_2(\theta)),\theta\right) - \frac{k_1}{u_1(\theta)}\frac{dC_1\left(s({\scriptstyle \ubar{r}}_2(\theta)),\theta\right)}{d\theta}>0 & \quad \text{if } \; \theta\in [0,\theta_h)\\ 
			-\frac{k_1}{u_1(\theta)^2}\frac{du_1(\theta)}{d\theta}C_1\left(s(\theta),\theta\right) + \frac{k_1}{u_1(\theta)}\frac{d C_1(s(\theta),\theta)}{d\theta}<0 & \quad \text{if } \theta\in(\theta_l,0).
		\end{cases}
	\end{displaymath}
\end{proof}

\subsubsection{Proof of Theorem~\ref{th:advthm}}

\begin{step}\label{step:uniqueness}
	Adversarial equilibria are essentially unique.
\end{step}

\begin{proof}
	The solution of equation~(\ref{eq:swing}) is unique and depends only on the model's primitives $u_\dm(\theta)$, $f(\theta)$, $u_i(\theta)$, $\tau_i$, $k_i$, $C_i(r_i,\theta)$, for $i\in\{1,2\}$. Therefore, for every $r\in[\ubar r_2(0),\bar r_1(0)]$, the swing report $s(r)$ is the same in every AE. It follows that the truthful cutoffs $\theta_l$ and $\theta_h$, and the senders' reporting strategies $\phi_j(\theta)$ and supports $S_j(\theta)$, $j\in\{1,2\}$, are also the same in all AE. Thus, all AE are strategy- and outcome-equivalent.
\end{proof}

To study whether there exist adversarial equilibria with unprejudiced beliefs I apply the following definition, which is adapted from \cite{bagwell1991} to accommodate non-degenerate mixed strategies.\footnote{Definition~\ref{def:unprejudiced}, introduced by \cite{vida2021} and used in Section~\ref{sec:fre}, is a weaker version of Definition~\ref{def:unprejudiced2}. Lemma~\ref{lemma:epsilonunprejudiced} applies to unprejudiced beliefs as in both definitions.} 
\begin{definition}\label{def:unprejudiced2}
	Given senders' strategies $\phi_j$, the decision maker's posterior beliefs $p$ are unprejudiced if, for every off path pair of reports $(r_1,r_2)$ such that $\phi_j(r_j,\theta')>0$ for some $j\in\{1,2\}$ and $\theta'\in\Theta$, we have that $p(\theta'' \,|\, r_1,r_2)>0$ if and only if there is a sender $i\in\{1,2\}$ such that $\phi_i(r_i,\theta'')>0$.
\end{definition}

\begin{Alemma}\label{lemma:unprejudicedDE}
	There exist adversarial equilibria with unprejudiced beliefs.
\end{Alemma}

\begin{proof}
	Consider an AE and an off path pair of reports $(r_1,r_2)$. By Steps~\ref{step:strategies}, \ref{step:supports}, and \ref{step:atoms}, and by Lemma~\ref{lemma:supportscutoffs}, we obtain that the only pair of reports such that $\phi_j(r_j,\theta)=0$ for all $\theta\in\Theta$ and $j\in\{1,2\}$ is $(0,0)$. For every other off path pair of reports, there is always a sender $i$ such that $\phi_i(r_i,\theta)>0$ for some $\theta\in\Theta$. There are three types of off path pairs of reports that need to be considered: those that violate Lemma~\ref{lemma:monot}, such as when $r_1>r_2$; those that violate Step~\ref{step:supports}, such as when $r_1>s(\ubar r_2(r_2))$; those that violate Lemma~\ref{lemma:supportscutoffs}, such as when $r_1\neq r_2$ for some $(r_1,r_2)\notin (\theta_l,\theta_h)^2$. 
	
	For beliefs to be unprejudiced, Definition~\ref{def:unprejudiced2} requires that for every such off path pair of reports we have that $p(\theta'' \,|\, r_1,r_2)>0$ if and only if there is a sender $i\in\{1,2\}$ such that $\phi_i(r_i,\theta'')>0$. Since $p(\theta'' \,|\, r_1,r_2)$ can be arbitrarily small, I can focus on the decision maker's beliefs that only one sender is deviating. Specifically, I will consider some posterior beliefs $p'$ which, given an off path pair of reports $(r_1,r_2)$ and provided that $\phi_j(r_j,\theta)>0$ for some $\theta\in\Theta$ and $j\in\{1,2\}$, rationalize deviations as originating with certainty from one specific sender $i$. If there is an AE with such posterior beliefs $p'$, then there exists an AE with posterior beliefs $p''$ (e.g., a small perturbation of $p'$) that satisfy Definition~\ref{def:unprejudiced2} (and hence also Definition~\ref{def:unprejudiced}).

	First, if $0\leq r_1<r_2$ (resp.~$r_1<r_2\leq 0$), then set $p'$ such that the decision maker believes that sender $1$ (resp.~$2$) is the deviator. Given the equilibrium strategies, it must be that sender~$2$ ($1$) is reporting truthfully, and therefore $p'$ leads to $\beta(r_1,r_2)=\p$ ($\n$). Consider now the case where $r_1<0<r_2$, and set $p'$ such that the decision maker believes that only sender~1 (or 2) is deviating. Therefore, it must be that sender~2 (1) is reporting truthfully, and therefore $\beta(r_1,r_2)=\p$ ($\n$). 
	Second, consider an off path pair of reports such that, for $x\geq 0$, we have that $r_2>x$ and $r_1\geq s(\ubar r_2 (x))$ (resp.~$r_1<y\leq 0$ and $r_2\leq s(\bar r_1(y))$). If the decision maker believes that sender $1$ ($2$) is the deviator, then it must be that $\theta=r_2$ ($\theta=r_1$) and therefore $\beta(r_1,r_2)=\p$ ($\n$). Finally, consider an off path pair $(r_1,r_2)\notin (\theta_l,\theta_h)^2$ with $r_1\neq r_2$. If both $r_1,r_2 \geq \theta_h$ (resp.~$r_1,r_2 \leq \theta_l$), then, by inferring that only one sender is deviating, the decision maker believes that $\theta\geq \theta_h>0$ ($\theta\leq\theta_l<0$) and selects $\beta(r_1,r_2)=\p$ ($\n$). 
	
	Since $p'$ is consistent with conditions~\eqref{eq:c} and \eqref{eq:d}, and since given $p'$ no sender is better off deviating from the prescribed equilibrium strategies, it follows that there are AE with unprejudiced beliefs as defined in Definitions~\ref{def:unprejudiced} and \ref{def:unprejudiced2}.
\end{proof}

The next corollary confirms that there exist adversarial equilibria supported by unprejudiced beliefs (as in both  Definition~\ref{def:unprejudiced} and \ref{def:unprejudiced2}) that are also $\varepsilon$-robust.\footnote{Since $\varepsilon$-robustness implies unprejudiced beliefs, it would be sufficient to show that there exist adversarial equilibria that are $\varepsilon$-robust. Corollary~\ref{step:robustDE} simply remarks that the two refinements are different.} 

\begin{step}\label{step:robustDE}
	There are adversarial equilibria with unprejudiced beliefs that are also $\varepsilon$-robust.
\end{step}

\begin{proof}
	Consider an $\varepsilon$-perturbed game with sequence $\varepsilon^n$ and full support distributions $\hat G=(\hat G_1, \hat G_2)$ such that $\hat g_1(r_1)\approx 0$ for all $r_1<0$ and $\hat g_2(r_2)\approx 0$ for all $r_2>0$. This means that it is relatively unlikely that the decision maker will misinterpret the report of sender~1 (resp. 2) to be negative (resp. positive). By equation \eqref{eq:epsilonbayes}, the limit beliefs $\hat p_{0^+}$ induced by the strategies of an AE after the decision maker observes a pair of reports $(r_1,r_2)$ such that $0\leq r_1<r_2$ are
	\begin{displaymath}
		\hat p_{0^+}(\theta\,|\,r_1\geq0,r_2>0)
		\approx  f(\theta)\frac{\delta(r_2 - \theta)\alpha_2(\theta)}{f(r_2)\alpha_2(r_2)}. 
	\end{displaymath}
	Therefore, the CDF $\hat P_{0^+}=\int \hat p_{0^+}(\theta\,|\,r_1,r_2)d\theta$ is such that
	\begin{displaymath}
		\hat P_{0^+}(\theta\,|\,r_1\geq 0,r_2>0) \approx
		\begin{cases}
			0 & \quad \text{if } \theta< r_2\\ 
			1 & \quad \text{if } \theta\geq r_2.
		\end{cases}
	\end{displaymath}
	
	As $\varepsilon^n \to 0^+$ and for every off path pair of reports that are both positive, the decision maker is almost sure that the realized state coincides with the report of sender~2. Similarly, we obtain that $\hat P_{0^+}(r_1<0,r_2\leq 0)\approx 0$ for all $\theta<r_1$ and $\approx 1$ otherwise, and by Lemma~\ref{lemma:epsilonunprejudiced} we have that $\hat p_{0^+}(\theta\,|\,r_1<0,r_2>0)>0$ only for $\theta\in\{r_1,r_2\}$. Therefore, the limit beliefs $\hat p_{0^+}$ are arbitrarily close to the posterior beliefs $p'$ in the proof of Lemma~\ref{lemma:unprejudicedDE}, and therefore can support an AE. Since an AE is also a PBE, by Lemma~\ref{lemma:epsilonunprejudiced} we obtain that some adversarial equilibria with unprejudiced beliefs are $\varepsilon$-robust.
\end{proof}

\begin{step}\label{step:existence}
	An adversarial equilibrium always exists.
\end{step}

\begin{proof}
	Given strategies $\phi_j(r_j,\theta)=\delta(r_j-\theta)\alpha_j(\theta)+\psi_j(r_j,\theta)$ as in Steps~\ref{step:strategies} and \ref{step:atoms}, with support $S_j(\theta)$ as in Step~\ref{step:supports}, posterior beliefs $p(\theta\,|\,r_1,r_2)$ are such that the swing report function $s(r)$ is as in Step~\ref{step:swing}. Given $s(r)$, strategies $\phi_j(r_j,\theta)$ are optimal by construction, and therefore no sender $j\in\{1,2\}$ is better off deviating from $\phi_j(r_j,\theta)$. Therefore, for every primitive of the model satisfying the conditions outlined in Section~\ref{sec:model}, there must exist an adversarial equilibrium as defined by Definition~\ref{def:directeqa}.
\end{proof}

\begin{proof}[{\bf Proof of Theorem~\ref{th:advthm}.}]
	The proof follows directly from Steps~\ref{step:uniqueness}--\ref{step:existence}. Part iii) is implied by Steps~\ref{step:uniqueness} and \ref{step:robustDE}.
\end{proof}

\section{Supplementary Appendix}\label{sec:suppapp}
\subsection{Example and Extensions}\label{app:example}

\symmetric*

\begin{proof}
	The proof follows directly from Step~\ref{step:swing}: consider a symmetric environment and suppose that $s(r)=-r$. Given a report $r\in(0,\bar r_1(0))$, the interval of integration in~\eqref{eq:swing} has $\max\{-r,\bar r_1^{-1}(r)\}=-\min\{r,\ubar r_2^{-1}(-r)\}$. Since the integrand in~\eqref{eq:swing} is symmetric around zero, we obtain that $G_s(r,-r)=0$, confirming that indeed $s(r)=-r$.
\end{proof}

\coalition*

\begin{proof}
	For a proof, see the appendix in \cite{vaccari2022efficient}.
\end{proof}

\silence*

\begin{proof}
	Consider the model's extension where senders can withhold information at a cost $c_j>0$, $j\in\{1,2\}$. Suppose there exists a REE where information withholding takes place on the equilibrium path. In a REE, given the senders' strategies, the decision maker selects alternative $\p$ when $\theta\geq 0$, and alternative $\n$ otherwise. Since both silence and misreporting are costly, the equilibrium must be such that sender~2 reports truthfully when $\theta\geq 0$, and sender~1 reports truthfully when $\theta<0$. As a result, information withholding is either performed by sender~1 in some $\theta\geq 0$ and/or by sender~2 in some $\theta<0$. 
	
	Suppose there is a REE where $r_2=\theta$ for all $\theta\geq 0$, while $r_1=\varnothing$ for some $\theta'\geq 0$ (the proof is similar for the other case). To sustain this REE, beliefs must assign $\beta(\theta',\theta')=\n$ to all $\theta'\geq 0$ such that $\rho_1(\theta')=\varnothing$. An argument similar to that in the proof of Proposition~\ref{prop:nonREE} shows that this cannot be a REE: sender~1 can deviate by reporting truthfully when the strategy prescribes information withholding. From the pair of reports $(\theta',\theta')$, the decision maker infers that sender~1 has deviated from the equilibrium strategy because sender~2 delivers $r_2=\theta'$ only in state $\theta'$ whereas sender~1 never delivers $\theta'$ on path. After learning from sender~2 that $\theta'\geq 0$, the decision maker must select $\p$, contradicting that $\beta(\theta',\theta')=\n$. As a result, there are no unprejudiced REE when silence is costly for both senders.
	
	Consider now a model's extension where sender~2 can withhold information at no cost, $c_2=0$, and focus on the REE with the following strategies: sender~1 plays the same reporting strategy as in Figure~\ref{fig:fre} (Section~\ref{sec:fre}), whereas sender~2 always stays silent (the proof for the case $c_1=0$ is similar). The beliefs of the decision maker are such that, no matter what sender~2 reports, she selects $\p$ when $r_1\geq \bar r_1(0)$, and selects $\n$ otherwise. Sender~2 cannot gain by deviating from his withholding strategy, for the decision maker would anyway ignore his reports.\footnote{Costless silence effectively restores the possibility of babbling (cfr. Lemma~\ref{lemma:babbling}).} The off path beliefs supporting this REE do not suffer from the problems outlined in Section~\ref{sec:fre}. To see this, consider the out-of-equilibrium contingency where the decision maker observes the pair of reports $(r_1',\varnothing)$ for some $r_1'\in[0,\bar r_1(0))$. Since sender~2 always stays silent, such an off path pair of reports cannot be originating from a double deviation. The decision maker would understand that sender~1 is the deviator but, differently from the FRE considered in Section~\ref{sec:fre}, now she cannot learn the realized state from sender~2. The decision maker can reply to any report of sender~1 that is lower than $\bar r_1(0)$ by selecting $\n$, under the belief that such a report may have been delivered in some negative state. Individual deviations by sender~2 are always detectable, and the state's sign is always revealed by sender~1's equilibrium reporting strategy. As a result, this REE is unprejudiced.
\end{proof}


\uncertain*

\begin{proof}
	In every equilibrium of  $\Gamma'$, misreporting occurs in some state outside $\left(\ubar t,\bar t\right)$. Suppose by way of contradiction that $\rho_j(\theta)=\theta$ for all $\theta\notin\left(\ubar t,\bar t\right)$ and $j\in\{1,2\}$. By Lemma~\ref{lemma:monot}, senders play monotonic reporting strategies\footnote{Since Lemma~1 follows from \eqref{eq:m}, it naturally applies to this variation of the baseline model as well.} in equilibrium, i.e., $\sup S_2(\theta)\leq \theta \leq \inf S_1(\theta)$. Therefore, it must be that $\beta\left(\bar t, \bar t\right)=\p$ and, to discourage deviations, beliefs must be such that $\beta\left(\bar t,\ubar t\right)=\p$. Since $\bar t < \bar r_1(\ubar t)$, sender~1 can profitably deviate by reporting $r_1=\bar t$ when $\theta=\ubar t$, contradicting the supposition.
	
	Next, consider a REE of $\Gamma'$. By definition of REE, it must be that for every $\theta\in\Theta$ and $r_j\in S_j(\theta)$ the decision maker selects $\beta(r_1,r_2)=\p$ if $\theta\geq\tau_{dm}$, and selects $\beta(r_1,r_2)=\n$ otherwise. When $\theta\geq\bar t$, the decision maker selects $\p$ on the equilibrium path, as $\bar t\geq \tau_{dm}$ surely. Consequently, when $\theta\geq\bar t$, truth-telling is strictly dominant for sender~2, who instead prefers action $\n$. For a similar reason, sender~1 reports truthfully in every state $\theta<\ubar t$. Since in every equilibrium there must be some misreporting, it has to be the case that sender~1 misreports in some state $\theta\geq \bar t$, sender~2 misreports in some $\theta<\ubar t$, or both.
	
	Suppose that, in this REE, sender~1 misreports in some state $\theta'\geq \bar t$. As argued before, sender~2 reports truthfully in that state, i.e., $\rho_2(\theta')=\theta'\neq \rho_1(\theta')$. By Lemma~\ref{lemma:monot}, in equilibrium sender~2 never delivers reports that are strictly higher than the realized state. That is, $\sup S_2(\theta)\leq \theta$. As a result, the report $r_2=\theta'$ is never delivered by sender~2 when the state is lower than $\theta'$. Since sender~2 reports truthfully in every $\theta\geq \bar t$, the report $r_2=\theta'$ is on path only when the state is $\theta'$. By contrast, the pair of reports $(\theta',\theta')$ is off path because sender~1 misreports in state $\theta'$. To sustain this REE, beliefs must be such that $\beta(\theta',\theta')=\n$, for otherwise sender~1 would prefer to not misreport in state $\theta'$. However, these off path beliefs have the same issues discussed in Section~\ref{sec:fre}: if we require the decision maker to conjecture deviations as originating from one sender only (whenever possible), then upon observing the pair $(\theta',\theta')$ she should infer that sender~1 performed the deviation. For otherwise, the decision maker must believe that sender~2 purposefully deviated from the prescribed equilibrium strategies by delivering a strictly dominated report. Given that sender~2 truthfully reports $\theta'$ only when the state is $\theta'$, the decision maker learns the true state after observing the off path pair of reports $(\theta',\theta')$. Consequently, she selects action $\p$ because $\theta'\geq\bar t$, contradicting that $\beta(\theta',\theta')=\n$. A similar argument applies to REE where sender~2 misreports in some state $\theta''< \ubar t$. As a result, if there is a REE of $\Gamma'$, then it is not unprejudiced.
\end{proof}

\subsection{An Example of Adversarial Equilibrium}\label{sec:symmexample}

The following example illustrates players' strategies in an adversarial equilibrium. Consider the following payoff structure: when the decision maker selects the positive alternative, sender~1 obtains a payoff of 1 and sender~2 obtains a payoff of -1; when selecting the positive alternative, the decision maker obtains a payoff of 1 if the state is positive, and a payoff of -1 if the state is negative. All players get a payoff of zero when the decision maker selects the negative alternative. Senders have no common interest: their thresholds are $\tau_1=-\infty$ and $\tau_2=+\infty$. Senders incur misreporting costs according to the function 
\[
k_j C_j(r_j,\theta)= |r_j-\theta|.
\]
From \eqref{eq:reach1} and \eqref{eq:reach2}, the upper reach of sender~1 in state $\theta$ is $\bar r_1(\theta)=\theta+1$, and the lower reach of sender~2 in state $\theta$ is $\ubar r_2(\theta)=\theta-1$. The state space is the real line, and the state is distributed according to a full support density function which is symmetric around zero.

Before proceeding, recall that in an adversarial equilibrium senders play monotonic reporting strategies: sender~1 delivers reports that are higher than or equal to the realized state; sender~2 delivers reports that are lower than or equal to the realized state (Lemma~\ref{lemma:monot}). As noted before, when both senders claim that the state is positive (resp. negative), the decision maker expects the state to be positive (resp. negative).

Suppose that, in case of conflicting reports, the decision maker assigns the same weight to the reports of the two senders. Specifically, when senders' reports are such that $-r_2\leq 0 \leq r_1$, the decision maker behaves as if she expects the state to be the average of the two reports. If this average is greater than zero, then she selects the positive alternative; otherwise, she selects the negative alternative. Alternatively, we can say that the decision maker follows the advice delivered by the most ``extreme'' report, that is, the report that is larger in absolute value.

Given the decision maker's behaviour, senders always report truthfully in sufficiently extreme states. For example, suppose that the realized state is $\theta=1/2$. Because of report monotonicity, sender~1 delivers reports that are higher than or equal to $1/2$. To persuade the decision maker in selecting the negative alternative, sender~2 must deliver a report that is lower than $-1/2$. Such a report is also lower than sender~2's reach, $\ubar r_2(1/2)=-1/2$, and thus he cannot profitably persuade the decision maker. As a result, sender~2 saves on costs by reporting truthfully, and sender~1 replies by following suit. This holds true in every state higher than $1/2$ and, by symmetry, in every state lower than $-1/2$. Full revelation via truthful reporting always takes place in sufficiently extreme states.

Suppose now that the realized state takes the more moderate value of $\theta=1/4$. To appreciate senders' behavior in this case, it is important to first understand which reports are never delivered in an adversarial equilibrium. For example, sender~2 never delivers reports that are not ``strong enough'' to persuade the decision maker. Recall that sender~1 never claims that the state is lower than its true value, and so $r_1\geq 1/4$. Therefore, sender~2 never misreports between $-1/4$ and $1/4$, as doing so would be costly but fruitless. Moreover, sender~2 does not deliver reports that are exceedingly low, as they would be too expensive. Specifically, sender~2 never delivers reports that are lower than his own reach, $\ubar r_2(1/4)=-3/4$. As a result, sender~1 does not need to deliver reports that are higher than $3/4$, even though his reach is $\bar r_1(1/4)=5/4$. 

We can focus now on what senders do deliver in an adversarial equilibrium, and how often. When $\theta=1/4$, sender~2 has to deliver at least $r_2<-1/4$ to persuade the decision maker. Differently than before, such a report is now within sender~2's reach, as $\ubar r_2(1/4)=-3/4$. To prevent sender~2 from persuading the decision maker, sender~1 must misreport as well. However, now senders' cannot play pure strategies: when sender~1 reports truthfully, sender~2 wants to misreport; when sender~2 misreports, sender~1 can misreport to the point that sender~2 is better off by reporting truthfully. However, when sender~2 reports truthfully, sender~1 wants to report truthfully as well, taking us back to square one. With probability $1/2$ each sender reports truthfully, and misreport otherwise. When misreporting, sender~1 is equally likely to deliver any report between $1/4$ and $3/4$; sender~2 is equally likely to deliver any report between $-3/4$ and $-1/4$.

The strategies in this adversarial equilibrium are as follows: in extreme states, when $\theta\notin(-1/2,1/2)$, senders always report truthfully; in moderate states, when $\theta\in(-1/2,1/2)$, each sender reports truthfully with probability $2|\theta|$; otherwise, they misreport. Conditional on misreporting when the state is positive, sender~1 is equally likely to deliver any report in the set $(\theta,1-\theta]$, and sender~2 is equally likely to deliver any report in the set $[\theta-1,-\theta]$. When the state is negative, sender~1 is equally likely to deliver any report in the set $[-\theta,\theta+1]$, and sender~2 is equally likely to deliver any report in the set $[-\theta-1,\theta)$. Given senders' symmetric strategies, when reports have opposed signs the decision maker behaves as if she expects the state to be the average of the two reports; when reports have the same sign, the decision maker learns that the state is equal to the report that is closer to zero.


\addcontentsline{toc}{section}{References}
\bibliographystyle{apacite}
\bibliography{biblio_compsig}
\end{document}